\documentclass[11pt, reqno]{amsart}
\usepackage{amssymb}
\usepackage[utf8]{inputenc}
\usepackage{graphicx}
\usepackage{xcolor}
\usepackage{enumerate}
\usepackage[pagebackref, colorlinks = true, linkcolor = blue, urlcolor  = blue, citecolor = red]{hyperref}
\usepackage[margin=1in]{geometry}
\usepackage{framed}
\usepackage{enumitem}
\usepackage{multirow}

\makeatletter
\newcommand{\mylabel}[2]{#2\def\@currentlabel{#2}\label{#1}}

\newtheorem{theorem}{Theorem}[section]
\newtheorem*{theorem*}{Theorem}
\newtheorem{definition}[theorem]{Definition}
\newtheorem{proposition}[theorem]{Proposition}

\newtheorem{lemma}[theorem]{Lemma}
\newtheorem{remark}[theorem]{Remark}
\newtheorem{example}[theorem]{Example}

\begin{document}

\title{On bipartite unitary matrices generating subalgebra--preserving quantum operations}

\author{Tristan Benoist}
\address{CNRS, Laboratoire de Physique Th\'eorique, Toulouse, France}
\email{tristan.benoist@irsamc.ups-tlse.fr}

\author{Ion Nechita}
\address{CNRS, Laboratoire de Physique Th\'eorique, Toulouse, France}
\email{nechita@irsamc.ups-tlse.fr}

\subjclass[2000]{}
\keywords{}

\begin{abstract}
We study the structure of bipartite unitary operators which generate via the Stinespring dilation theorem, quantum operations preserving some given matrix algebra, independently of the ancilla state. We characterize completely the unitary operators preserving diagonal, block-diagonal, and tensor product algebras. Some unexpected connections with the theory of quantum Latin squares are explored, and we introduce and study a Sinkhorn--like algorithm used to randomly generate quantum Latin squares.
\end{abstract}

\date{\today}

\maketitle

\tableofcontents

\section{Introduction}
States of finite dimensional quantum systems are described by trace one, positive semidefinite matrices called \emph{density matrices}. Their evolution is given by completely positive trace preserving maps (CPTP), also called \emph{quantum channels}. For \emph{closed} systems, the Quantum channel consists in the left and right multiplication by a unitary matrix and its hermitian conjugate respectively. Quantum channels describing the evolution of a system in contact with an environment, the so--called \emph{open} quantum systems, are given by more general CPTP maps. The famous Stinespring dilation theorem \cite{sti} connects the mathematical definition of open system evolution quantum channels with their physical interpretation: the evolution of an open quantum system can be seen as the closed evolution of a [system + environment] bipartite system, followed by the discarding (or partial tracing) of the environment:
$$\rho\mapsto T(\rho) = [\operatorname{id} \otimes \operatorname{Tr}](U(\rho \otimes \beta) U^*).$$
Hence, the quantum evolution depends on two physical relevant parameters: the global evolution operator $U$ and the state of the environment subsystem $\beta$. We would like to identify the part played by the unitary interaction $U$ in the properties of the quantum channel $T$. Namely, we would like to characterize families of bipartite unitary matrices $U$ such that the quantum channels $T$ they generate have some prescribed properties, independently of the environment state $\beta$. In this work, we answer this question for channels preserving some special sub--algebras of states/observables. Different channel properties were studied under the same framework in \cite{dnp} (see also \cite{bar}). 

Part of the initial motivation of our study originates in the characterization of quantum channels preserving a set of \emph{pointer} states. These CPTP maps appear in the context of quantum non demolition measurements \cite{bb,bbb,bpe,adler,vanhandel,stockton}. The following proposition is a discrete time version of \cite[Theorem 3]{bpe}. It is obtained through the repeated application of Theorem \ref{thm:Zero_bloc_unitary}.
\begin{proposition}\label{prop:preserves-basis}
Let $\{e_i\}_{i=1}^n$ be a fixed orthonormal basis of $\mathbb C^n$, and $U \in \mathcal U_{nk}$ a bipartite unitary operator. The following statements are equivalent:
\begin{enumerate}
\item There exists a set of quantum states $\mathcal B$ which spans $M_k(\mathbb C)$ such that, for any state $\beta \in \mathcal B$, the quantum channel $T_{U,\beta}$ leaves invariant the basis states $e_ie_i^*$:
$$\forall \beta \in \mathcal B, \, \forall i \in [n], \qquad T_{U,\beta}(e_ie_i^*) = e_ie_i^*;$$
\item The operator $U$ is block--diagonal, i.e. there exists unitary operators $U_1, \ldots, U_n \in \mathcal U_k$ such that
$$U = \sum_{i=1}^n e_i e_i^* \otimes U_i.$$
\end{enumerate}
The implication (1)$\implies$(2) holds also if (1) is replaced by
\[\exists \beta>0, \forall i\in[n], \qquad T_{U,\beta}(e_ie_i^*)=e_ie_i^*.\]
\end{proposition}

The main goal of the current paper is to obtain similar characterization of the set of bipartite unitary operators giving quantum operations which preserve some structure on the input state or observable, such as diagonal matrices, tensor products, or other block-structures. One of our main results is a characterization of quantum channels creating no coherences between states of a prescribed basis. They therefore act essentially classically. More precisely we characterize bipartite unitary matrices generating CPTP maps preserving quantum states which are diagonal in the computation basis (see Theorem \ref{thm:diagonal-Schrodinger} for a precise statement):
\begin{theorem*}
A bipartite unitary operator $U$ acting on $\mathbb C^n \otimes \mathbb C^k$ generates quantum channels $T$ which, independently of the environment state $\beta$, leave invariant the diagonal subalgebra $D_n$ if and only if its blocks $\{U_{ij}\}_{i,j=1}^n$ are partial isometries with initial and final spaces $E_{ij}$ and $F_{ij}$ respectively satisfying the following conditions:
\begin{itemize}
\item For all rows $i$, the collection of final spaces $\{F_{ij}\}_j$ forms an orthogonal partition of $\mathbb C^k$
\item For all columns $j$, both collections of initial spaces $\{E_{ij}\}_i$ and of final spaces $\{F_{ij}\}_i$ form orthogonal partitions of $\mathbb C^k$.
\end{itemize}
\end{theorem*}

In the result above, the collection of final spaces of the partial isometries $\{F_{ij}\}_{ij}$, is required to partition the total space $\mathbb C^k$ simultaneously on every row and column. In the case where each space $F_{ij}$ is of dimension one, this is precisely the definition of \emph{quantum Latin squares} from \cite{mvi}. We analyze further these objects and we propose an algorithm to randomly sample such objects. We analyze numerically this algorithm (inspired by the classical Sinkhorn algorithm) and prove the convergence of a relaxed version of it.

The paper is organized as follows. In Section \ref{sec:channels-and-subalgebras} we set up the problem we are studying, giving the necessary definitions. In Section \ref{sec:partial-isometries} we gather some useful facts about partial isometries which shall be necessary later on. The next four sections represent the core of the paper and contain the results about, respectively, diagonal, block--diagonal, tensor product, and zero  algebras. The proofs for the preservation of diagonal or block diagonal algebra and tensor product algebra use incompatible techniques, we thus do not provide a characterization of bipartite unitary matrices generating quantum channels preserving direct sum of tensor product sub--algebras. Section \ref{sec:algorithms} contains some numerical recipes to sample from the sets of unitary operators considered in this paper. Finally, the Appendix contains the proof of the non--commutative Sinkhorn algorithm in the case where the blocks of the matrix are invertible.

\bigskip

\noindent \emph{Acknowledgments.} I.N.~would like to thank Teo Banica and Martin Idel for inspiring discussions around Sinkhorn's algorithm and its different generalizations. The authors' research has been supported by the ANR projects {RMTQIT}  {ANR-12-IS01-0001-01} and {StoQ} {ANR-14-CE25-0003-01}. Both authors acknowledge the hospitality of the Mathematical Physics Chair of the Technische Universit\"at M\"unchen, where this research was initiated. 

\section{Quantum channels and invariant subalgebras}
\label{sec:channels-and-subalgebras}

In this paper, we are interested in the structure of bipartite unitary operators $U \in \mathcal U_{nk}$, and of the quantum operations they generate via dilation results. To any such bipartite unitary operator and any trace one positive semidefinite matrix $\beta  \in M_k(\mathbb C)$, we associate two maps:
\begin{align}
\nonumber S_{U,\beta} :M_n(\mathbb C) &\to M_n(\mathbb C)\\
\label{eq:def-UCP} X &\mapsto [\operatorname{id} \otimes \operatorname{Tr}](U^*(X \otimes I_k)U(I_n \otimes \beta))
\end{align}
and
\begin{align}
\nonumber T_{U,\beta} :M_n(\mathbb C) &\to M_n(\mathbb C)\\
\label{eq:def-TPCP} X &\mapsto [\operatorname{id} \otimes \operatorname{Tr}](U(X \otimes \beta)U^*)
\end{align}

If the matrix $\beta$ is a quantum state, i.e. $\beta \in M_k^{1,+}(\mathbb C)$, the map $S_{U,\beta}$ is unital and completely positive (UCP), while the map $T_{U,\beta}$ is trace preserving and completely positive (TPCP). Moreover, a direct calculation shows that the maps above are dual with respect to the Hilbert--Schmidt scalar product:
$$\forall X,Y \in M_n(\mathbb C), \qquad \langle S_{U,\beta}(X), Y \rangle = \langle X, T_{U,\beta}(Y) \rangle.$$
We refer the reader to \cite[Chapter 6]{pau} for a mathematical introduction to CP maps, and to \cite[Chapter 8]{nch} or \cite{wol} for a quantum information theory perspective.

Our ultimate goal will be to characterize bipartite unitary operators $U$ with the property that \emph{all} the UCP maps $S_{U,\beta}$ (respectively all quantum channels $T_{U,\beta}$)  preserve a $C^*$--algebra 
\begin{equation}\label{eq:C-star-subalgebra}
\mathcal A = 0_{d_0} \oplus \bigoplus_{i=1}^N \left( M_{d_i}(\mathbb C) \otimes I_{r_i} \right).
\end{equation}
Note that the above expression is, up to a global unitary rotation, the most general form of a $C^*$--subalgebra of $M_n(\mathbb C)$, see \cite[Theorem I.10.8]{dav}. The diagonal subalgebra corresponds to the case $d_0 = 0$, $d_i = r_i = 1$, for $i \in [N]$ (for an integer $n$, we write $[n]:=\{1, 2, \ldots, n\}$). More precisely, we shall be interested in the sets
\begin{align*}
\mathcal U^{(\mathcal A, H)} &:= \{U \in \mathcal U_{nk} \, : \, \forall \beta \in M_k^{1,+}(\mathbb C), \, S_{U,\beta}(\mathcal A) \subseteq \mathcal A\} \\
\mathcal U^{(\mathcal A, S)} &:= \{U \in \mathcal U_{nk} \, : \, \forall \beta \in M_k^{1,+}(\mathbb C), \, T_{U,\beta}(\mathcal A) \subseteq \mathcal A\} 
\end{align*} 
of bipartite unitary operators leaving the subalgebra $\mathcal A$ invariant in the Heisenberg, respectively the Schr\"odinger picture. We shall also be interested in comparing the sets above corresponding to two subalgebras $\mathcal A_1 \subseteq \mathcal A_2 \subseteq M_n(\mathbb C)$. In general, we have the following inclusion diagram

\renewcommand{\arraystretch}{1.5}
\begin{center}
\begin{tabular}{ccc}
$\{U \, : \, \forall \beta, \, S_{U,\beta}(\mathcal A_1) \subseteq \mathcal A_1\}$ & $\subseteq$  & $\{U \, : \, \forall \beta, \, S_{U,\beta}(\mathcal A_1) \subseteq \mathcal A_2\}$ \\
\rotatebox{90}{$\subseteq$} & & \rotatebox{90}{$\subseteq$}
 \\
$\{U \, : \, \forall \beta, \, S_{U,\beta}(\mathcal A_2) \subseteq \mathcal A_1\}$ & $\subseteq$  & $\{U \, : \, \forall \beta, \, S_{U,\beta}(\mathcal A_2) \subseteq \mathcal A_2\}$ 
\end{tabular}
\end{center}

Similar relations hold for channels in the Schr\"odinger picture. Note that the above diagram does not imply any obvious relations between the sets $\mathcal U^{(\mathcal A_1, H/S)}$ and  $\mathcal U^{(\mathcal A_2, H/S)}$; some special cases of algebra inclusions will be investigated in Sections \ref{sec:comparing-sum} and \ref{sec:comparing-tensor}, to show that these sets are, in general, not comparable.

\section{Structures arising from partial isometries}
\label{sec:partial-isometries}
Recall that a \emph{partial isometry} between two Hilbert spaces is a map such that its restriction to the orthogonal of its kernel is an isometry. More precisely, a matrix $R \in M_n(\mathbb C)$ is a partial isometry if there exists two subspaces $E$ and $F$, called respectively the \emph{initial} and the \emph{final subspace} of $R$, such that
\begin{enumerate}
\item $\ker R = E^\perp$;
\item $R:E \to F$ is an isometry.
\end{enumerate}
Partial isometries are characterized by the relation $R=RR^*R$ (or, equivalently, by the relation $R^*= R^*RR^*$). Let us remark right away that two important classes of operators, orthogonal projections and unitary operators, are partial isometries.

\begin{definition}
A family of vector subspaces $\{E_i\}_{i \in I}$ is said to form a \emph{partition} of $\mathbb C^k$ if
\begin{enumerate}
\item For all $i \neq j \in I$, we have $E_i \perp E_j$;
\item $\oplus_{i \in I} E_i = \mathbb C^k$.
\end{enumerate}
\end{definition}

\begin{definition}\label{def:bistochastic-PI}
An operator $R \in M_{nk}(\mathbb C) \cong M_n(M_k(\mathbb C))$ is called a \emph{matrix of partial isometries} if its blocks $R_{ij} \in M_k(\mathbb C)$ defined by
$$R = \sum_{i,j=1}^n e_ie_j^* \otimes R_{ij}$$
are partial isometries. Let $E_{ij}$ (resp.~$F_{ij}$) be the initial (resp.~final) spaces of the partial isometries $R_{ij}$. The matrix of partial isometries $R$ is said to be of \emph{type $(x,y,z,\ldots)$ }if, in addition, the subspaces $E_{ij},F_{ij}$ satisfy the orthogonality conditions $x,y,z,\ldots$ from the list below:
\begin{enumerate}[label=\textnormal{(\arabic*)}]
\item[\mylabel{it:row-initial}{(C1)}] For all $i \in [n]$, the subspaces $\{E_{ij}\}_{j \in [n]}$ form a partition of $\mathbb C^k$;
\item[\mylabel{it:row-final}{(C2)}] For all $i \in [n]$, the subspaces $\{F_{ij}\}_{j \in [n]}$ form a partition of $\mathbb C^k$;
\item[\mylabel{it:col-initial}{(C3)}] For all $j \in [n]$, the subspaces $\{E_{ij}\}_{i \in [n]}$ form a partition of $\mathbb C^k$;
\item[\mylabel{it:col-final}{(C4)}] For all $j \in [n]$, the subspaces $\{F_{ij}\}_{i \in [n]}$ form a partition of $\mathbb C^k$.
\end{enumerate}
\end{definition}

\begin{lemma}\label{lem:unitary-and-PI}
A matrix of partial isometries is unitary iff it is of type \ref{it:row-final},\ref{it:col-initial}.
\end{lemma}
\begin{proof}
Let $R$ be a matrix of partial isometries of type \ref{it:row-final},\ref{it:col-initial}, with blocks $R_{ij}$. 
We have 
$$RR^* = \sum_{i,j=1}^n e_ie_j^* \otimes \sum_{s=1}^n R_{is}R_{js}^*.$$
Consider now a general term $R_{is}R_{js}^*$ in the sum above. The map $R_{js}^*$ is also a partial isometry, having initial space $F_{js}$ and final space $E_{js}$. Since, for all $s$, the spaces $\{E_{rs}\}_{r \in [n]}$ form a partition of $\mathbb C^k$, we obtain $R_{is}R_{js}^* = \delta_{ij} P_{F_{is}}$ ($P_V$ denotes the orthogonal projection onto some subspace $V$). Hence 
$$RR^* = \sum_{i=1}^n e_ie_i^* \otimes \sum_{s=1}^n P_{F_{is}} = \sum_{i=1}^n e_ie_i^* \otimes I_k = I_{nk},$$
where we have used again the direct sum hypothesis for the partition $\{F_{is}\}_{s \in [n]}$.

Reciprocally, consider now a unitary operator $U$ and represent it by blocks (which are partial isometries)
$$U = \sum_{i,j=1}^n e_ie_j^* \otimes U_{ij}.$$
From the unitarity condition $UU^* = I_{nk}$, we get that, for all $i \in [n]$, 
$$I_k = \sum_{j=1}^n U_{ij}U_{ij}^* = \sum_{j=1}^n P_{F_{ij}}.$$
But a sum of selfadjoint projections is equal to the identity iff the subspaces of the projections form a partition of the total space. Hence, we obtain that, for all $i$, the subspaces $\{F_{ij}\}_{j \in [n]}$ form a partition of $\mathbb C^k$, which is condition \ref{it:row-final} from Definition \ref{def:bistochastic-PI}. Starting from the relation $U^*U=I_{nk}$, a similar reasoning produces condition \ref{it:col-initial} from the statement.
\end{proof}

The proofs of the upcoming sections are based on the following lemma, due to Cochran \cite{coc}. 

\begin{lemma}\label{lem:Cochran}
Let $A_1, \ldots, A_n \in M_k(\mathbb C)$ be operators satisfying 
\begin{equation}\sum_{i=1}^n A_i^*A_i = I.\label{eq:sum Id}\end{equation}
Then, 
\begin{equation}\sum_{i=1}^n \mathrm{rk}(A_i) \geq k,\label{eq:rank ineq}\end{equation}
with equality iff the $A_i$ are partial isometries with initial spaces $E_i$ forming a partition of $\mathbb C^k$.
\end{lemma}
\begin{proof}
The rank sum inequality \eqref{eq:rank ineq} is a consequence of the rank subadditivity and the equality ${\rm rk}(A^*A)={\rm rk}(A)$ for any matrix $A\in\mathcal M_k(\mathbb C)$.

The fact that if the $A_i$'s are partial isometries with initial spaces forming a partition of $\mathbb C^k$ then $\sum_{i=1}^n {\rm rk}(A_i)=k$ and $\sum_{i=1}^n A_i^*A_i=I$ follows directly from $A_i^*A_i=P_{E_i}$ and ${\rm rk}(A_i)=\dim(E_i)$ by the Rank--nullity Theorem.

Let us prove that $\sum_{i=1}^n {\rm rk}(A_i)=k$ and $\sum_{i=1}^n A_i^*A_i=I$ imply that the $A_i$'s are partial isometries with the initial spaces $E_i$ forming a partition of $\mathbb C^k$.

For general matrices the spaces $E_i$ are defined as the orthogonal complement of $\ker(A_i)$ in $\mathbb C^k$.  Let $E$ be the direct sum of the $E_i$'s: $E=\bigoplus_{i=1}^n E_i$. Assume $E\subset \mathbb C^k$ with a strict inclusion. Then $\exists x\in\mathbb C^k$ different from $0$ such that $\forall i\in[n], A_i x=0$. From \eqref{eq:sum Id}, it follows that 
\[\|x\|_2^2=\sum_{i=1}^n \langle A_i x, A_i x\rangle=0.\]
That contradicts the condition $x\neq0$ thus $E=\mathbb C^k$.

From \eqref{eq:sum Id} again, we have that for every $i\in[n]$, $A_i^*A_i\leq I$, hence $A_i^*A_i\leq P_{E_i}$ since $A_i^*A_i$ is $0$ on the orthogonal of $E_i$. It follows that $\sum_{i=1}^n P_{E_i}\geq I$ by summing over $i$. Assume this last inequality is strict. From the rank inequality \eqref{eq:rank ineq}, $\sum_{i=1}^n {\rm rk}(A_i)>k$ since ${\rm rk}(A_i)=\dim(E_i)$ by the Rank--nullity Theorem. Since we assumed that $\sum_{i=1}^n {\rm rk}(A_i)=k$ from the contradiction we obtain $\sum_{i=1}^n P_{E_i}= I$. Multiplying on the left and right by a specific $P_{E_j}$ we obtain $\sum_{i\neq j} P_{E_j} P_{E_i} P_{E_j}=0$. Since this sum is a sum of positive matrices equal to $0$, all the terms must be $0$. Hence, the mutual orthogonality of the $E_i$'s follows. 

We proved that the family $\{E_i\}_{i\in[n]}$ is a partition of $\mathbb C^k$. It remains to prove that the $A_i$'s are partial isometries.

From the mutual orthogonality of the $E_i$'s and the equality \eqref{eq:sum Id}, for any $j\in[n]$,
\[A_j\sum_{i=1}^nA_i^*A_i=A_jA_j^*A_j=A_j.\]
Hence the $A_i$'s are partial isometries.
\end{proof}

We prove next a more specialized version of Cochran's lemma.

\begin{lemma}\label{lem:partial-isometries}
Let $A_1, \ldots, A_n \in M_k(\mathbb C)$ such that 
\begin{align}
\label{eq:sum-eq-I} \sum_{i=1}^n A_i^*A_i &= I\\
\label{eq:eq-0}\forall i \neq j, \qquad A_i^*A_j &= 0.
\end{align}
Then, the matrices $A_i$ are partial isometries with the family of initial vector spaces $\{E_i\}$ and the family of final  vector spaces $\{F_i\}$ forming partitions of $\mathbb C^k$.
\end{lemma}
\begin{proof}
The final vector subspaces $F_i$ are the image of $\mathbb C^k$ by the respective $A_i$: $F_i=A_i\mathbb C^k$. The rank of $A_i$ is equal by definition to the dimension of $F_i$. From the definition of the vector spaces $F_i$, for any $x\in F_i$ and $y\in F_j$, there exists $x',y'\in \mathbb C^k$ such that $x=A_ix'$ and $y=A_jy'$. From \eqref{eq:eq-0}, if $i\neq j$,
\[\langle x,y\rangle=\langle A_ix',A_jy'\rangle=\langle x',A_i^*A_j y'\rangle=0.\]
It follows that the $F_i$'s are mutually orthogonal.

Since the $F_i$'s are mutually orthogonal, the dimension of the direct sum $F=\bigoplus_{i=1}^n F_i$ is equal to the sum of the dimensions of the $F_i$'s: $\dim(F)=\sum_{i=1}^n\dim(F_i)$. Since $F$ is a vector subspace of $\mathbb C^k$, $\dim(F)\leq k$. The equality and then the conclusion follow from Lemma \ref{lem:Cochran} and \eqref{eq:sum-eq-I}.
\end{proof}

\section{The diagonal algebra}

We start with the most simple case, that of the diagonal subalgebra
$$D_n: = \bigoplus_{a=1}^n M_1(\mathbb C).$$

From a physical perspective, this case is probably the most interesting one, and it is also very illuminating from a mathematical point of view. We shall treat separately the cases of CP maps which are unital (the Heisenberg picture) and trace preserving (quantum channels, the Schr\"odinger picture). Although the results in this section are special cases of the ones in Section \ref{sec:block-diagonal}, we state and prove them separately, in order to showcase the different ideas and techniques involved in the proof, which are more transparent in this case.

\subsection{The Heisenberg picture}

We are interested physical transformations of observables, that is in unital completely positive maps (UCP). More precisely, such maps can be written as
$$S_{U,\beta}(X) = [\mathrm{id} \otimes \mathrm{Tr}](U^*(X \otimes I_k)U(I_n \otimes \beta)),$$
where $U \in \mathcal U_{nk}$ is a unitary operator, $\beta \in M_k^{1,+}(\mathbb C)$ is a density matrix, and $X \in M_n^{sa}(\mathbb C)$ is the observable (Hermitian matrix) on which the map acts. 

\begin{theorem}\label{thm:diagonal-Heisenberg}
Let $\{e_i\}_{i=1}^n$ be a fixed basis of $\mathbb C^n$, $D_n$ the commutative algebra of matrices which are diagonal in the basis $\{e_i\}$, and $U \in \mathcal U_{nk}$ a bipartite unitary operator. The following statements are equivalent:
\begin{enumerate}
\item There exists a set of quantum states $\mathcal B$ which spans $M_k(\mathbb C)$ such that, for any state $\beta \in \mathcal B$, the UCP map $S_{U,\beta}$ leaves invariant diagonal matrices:
$$\forall \beta \in \mathcal B, \qquad S_{U,\beta}(D_n) \subseteq D_n;$$
\item In the basis $\{e_i\}$, the operator $U$ is a matrix of partial isometries (and thus it is of type \ref{it:row-final},\ref{it:col-initial}).
\end{enumerate}
\end{theorem}
\begin{proof}
The easy implication, $(2) \implies (1)$,  is proved by direct calculation. By our assumption, the blocks of $U$ are partial isometries, with initial and final spaces satisfying conditions \ref{it:row-final},\ref{it:col-initial} from Definition \ref{def:bistochastic-PI}
$$U = \sum_{i,j=1}^n e_i e_j^* \otimes U_{ij}.$$

For each basis element $e_s$, we have
\begin{align*}
S_{U,\beta}(e_se_s^*) &= \sum_{i_1,j_1,i_2,j_2 =1}^n [\operatorname{id} \otimes \operatorname{Tr}](e_{j_1}e_{i_1}^* \otimes U_{i_1j_1}^* \cdot e_s e_s^* \otimes I_k \cdot e_{i_2}e_{j_2}^* \otimes U_{i_2 j_2} \cdot I_n \otimes \beta) \\
&= \sum_{j_1,j_2 = 1}^n e_{j_1}e_{j_2}^* \operatorname{Tr}(U_{sj_1}^* U_{sj_2} \beta).
\end{align*}
Since the blocks $U_{ij}$ satisfy condition \ref{it:row-final} from Definition \ref{def:bistochastic-PI}, the operator $U_{sj_1}^* U_{sj_2}$ is null, unless $j_1=j_2$. Hence, the output $S_{U,\beta}(e_se_s^*) $ is a diagonal operator, proving the claim. 

Let us now show $(1) \implies (2)$. Going through the computations above, without assuming anything about the blocks of $U$, we see that the condition that $S_{U,\beta}$ preserves the diagonal implies that for all $s,j_1,j_2 \in [n]$ such that $j_1 \neq j_2$ and for all $\beta \in \mathcal B$, 
$$\operatorname{Tr}(U_{sj_1}^* U_{sj_2} \beta)= 0.$$
Since $\mathcal B$ spans the whole matrix algebra $M_k(\mathbb C)$, we obtain $U_{sj_1}^* U_{sj_2} =0$. In turn,
this implies $F_{sj_1}\perp F_{sj_2}$ for all $s,j_1,j_2 \in[n]$ such that $j_1\neq j_2$. Hence $F_s=\oplus_{j} F_{sj}$ is a direct sum of orthogonal subspaces of $\mathbb C^k$. In particular, $\dim F_s=\sum_j \dim F_{sj}\leq k$.

Since $U$ is unitary, we moreover have, for all $s \in [n]$, $\sum_j U_{sj}U^*_{sj}=I_k$. From Cochran's Lemma \ref{lem:Cochran} (with $A_j = U_{sj}^*$) we deduce $\dim F_s\geq k$. Hence $\sum_j{\rm rk}(U_{sj})=k$ and the $U_{sj}$'s are partial isometries of type \ref{it:row-final}.

Again using the unitarity of $U$, we get $\sum_i U_{ij}^*U_{ij}=I_k$ for all $j\in[n]$. Since the $U_{ij}$'s are partial isometries, the previous equality reads $\sum_{i} P_{E_{ij}}=I_k$ for all $j\in[n]$. It follows that, at fixed $j$, the initial spaces $\{E_{ij}\}_{i \in [n]}$ form a partition of $\mathbb C^k$, which is property \ref{it:col-initial} from Definition \ref{def:bistochastic-PI}. We conclude that the $U_{ij}$'s are partial isometries of type \ref{it:row-final},\ref{it:col-initial}, proving the claim.
\end{proof}
\begin{remark}
Let us consider the simplest case, where $n=k=2$, and the partial isometries appearing as the blocks of $U$ have unit rank. The $4 \times 4$ unitary operators which satisfy the conditions of the result above are of the form
$$U = \begin{bmatrix}
ab^* & a_\perp c^* \\
d b_\perp^* & d_\perp c_\perp^*
\end{bmatrix} \in  M_2(M_2(\mathbb C)),$$
for some vectors $a, a_\perp, b,b_\perp, c,c_\perp, d,d_\perp \in \mathbb C^2$ satisfying the indicated orthogonality relations (we denote by $x_\perp$ a fixed unit vector which is orthogonal to $x \in \mathbb C^2$). The UCP map $S_{U,\beta}$ acts as follows on diagonal operators:
$$S_{U,\beta}\left(\begin{bmatrix} x & 0 \\ 0 & y \end{bmatrix} \right) = \begin{bmatrix} x \langle a, \beta a \rangle + y \langle a_\perp ,\beta a_\perp \rangle & 0 \\ 0 & x \langle d, \beta d \rangle + y \langle d_\perp ,\beta d_\perp \rangle \end{bmatrix} .$$
\end{remark}

\begin{example}\label{ex:id-flip}
The identity and the flip operator
$$F_n = \sum_{i,j=1}^n e_ie_j^* \otimes e_je_i^* \in \mathcal U_{n^2}$$
satisfy the conditions in Theorem \ref{thm:diagonal-Heisenberg}, since
\begin{align*} 
S_{I_{nk}, \beta}(X) &= \operatorname{Tr}(\beta) X \\
S_{F_{n}, \beta}(X) &= \operatorname{Tr}(X\beta) I.
\end{align*}
\end{example}

\subsection{The Schr\"odinger picture}

Let $D_n \subset M_n(\mathbb C)$ the set of diagonal $n \times n$ matrices, in some basis $\{e_i\}_{i=1}^n$ of $\mathbb C^n$, which we assume fixed.

The main result of this section is the following theorem, which characterizes unitary matrices preserving, independently of the ancilla state, the diagonal subalgebra of the system space.

\begin{theorem}\label{thm:diagonal-Schrodinger}
Let $\{e_i\}_{i=1}^n$ be a fixed basis of $\mathbb C^n$, $D_n$ the commutative algebra of matrices which are diagonal in the basis $\{e_i\}$ and $U \in \mathcal U_{nk}$ a bipartite unitary operator. The following statements are equivalent:
\begin{enumerate}
\item There exists a set of quantum states $\mathcal B$ which spans $M_k(\mathbb C)$ such that, for any state $\beta \in \mathcal B$, the quantum channel $T_{U,\beta}$ leaves invariant diagonal matrices:
$$\forall \beta \in \mathcal B, \qquad T_{U,\beta}(D_n) \subseteq D_n;$$
\item In the basis $\{e_i\}$, the operator $U$ is a matrix of partial isometries of type \ref{it:row-final},\ref{it:col-initial},\ref{it:col-final}.
\end{enumerate}
When this is the case, for all quantum states $\beta \in M_k^{1,+}(\mathbb C)$, the quantum channel $T_{U,\beta}$ acts on diagonal matrices as a Markov chain with transition probabilities $(P_{ij})_{i,j \in [n]}$ given by
$$P_{ij} = \operatorname{Tr}(P_{E_{ji}} \beta),$$
where $E_{ji}$ is the initial space of the partial isometry corresponding to the block $(j,i)$ of $U$.
\end{theorem}
\begin{proof}
The proof follows more or less the same steps as the proof of Theorem \ref{thm:diagonal-Heisenberg}. Let us start with the easier implication (2) $\implies$ (1). Consider a basis element $e_s$ and compute
\begin{equation}\label{eq:T-basis-state}
T_{U,\beta}(e_s e_s^*) = \sum_{i,j=1}^n e_i e_j^* \operatorname{Tr}(U_{is} \beta U_{js}^*),
\end{equation}
where $U_{ij} \in M_k(\mathbb C)$ are the blocks of $U$:
$$U=\sum_{i,j=1}^n e_i e_j^* \otimes U_{ij}.$$
Since $U$ is of type \ref{it:col-initial}, we have $U_{js}^*U_{is} = \delta_{ij}P_{E_{is}}$, and thus
$$T_{U,\beta}(e_s e_s^*) = \sum_{i=1}^n e_i e_i^* \operatorname{Tr}(P_{E_{is}} \beta),$$
proving the claim.

\medskip

We prove now the second implication, (1) $\implies$ (2). Using equation \eqref{eq:T-basis-state}, we obtain that $T_{U,\beta}(e_s e_s^*)$ is a diagonal element iff
$$\forall \beta \in \mathcal B, \, \forall i \neq j \in [n], \qquad \operatorname{Tr}(U_{is} \beta U_{js}^*) = 0.$$
Since the above relation holds for a set of matrices $\beta$ which spans the full matrix algebra, we conclude that
\begin{equation}\label{eq:U-js-U-is}
\forall s \in [n], \, \forall i \neq j \in [n], \qquad U_{js}^*U_{is}=0.
\end{equation}
Using the unitarity condition $U \in \mathcal U_{nk}$, we also have
$$\sum_{j=1}^n U_{js}^*U_{js} = I_k.$$ 
We now apply, for every $s$, Lemma \ref{lem:partial-isometries} to conclude that all the blocks $U_{ij}$ of $U$ are partial isometries having the property that, for each column, both the initial and the final spaces of the partial isometries form a partition of $\mathbb C^k$ (these are the conditions \ref{it:col-initial} and \ref{it:col-final} in Definition \ref{def:bistochastic-PI}). Condition \ref{it:row-final} follows from Lemma \ref{lem:unitary-and-PI}. 
\end{proof}

\begin{remark}
The condition \ref{it:row-initial} from Definition \ref{def:bistochastic-PI} is not needed for the result above to hold, as the following example shows. Consider three unit vectors $a,b,c \in \mathbb C^2$ with the property that $\langle b,c\rangle \neq 0$, and let 
$$U = \begin{bmatrix}
ab^* & a_\perp c^* \\
a_\perp b_\perp^* & a c_\perp^*
\end{bmatrix} \in  M_2(M_2(\mathbb C)).$$
By direct calculation, one can show that $U$ is a unitary matrix. Notice that $U$ is a matrix of partial isometries satisfying conditions \ref{it:row-final},\ref{it:col-initial},\ref{it:col-final} from Definition \ref{def:bistochastic-PI}, but failing condition \ref{it:row-initial}. Moreover, for a given qubit density matrix $\beta \in M_2^{1,+}(\mathbb C)$, the channel $T_{U,\beta}$ acts on diagonal operators as follows:
\begin{align*}
T_{U,\beta}(e_1e_1^*) &= \langle b, \beta b \rangle e_1e_1^* + \langle b_\perp, \beta b_\perp \rangle e_2e_2^*\\
T_{U,\beta}(e_2e_2^*) &= \langle c, \beta c \rangle e_1e_1^* + \langle c_\perp, \beta c_\perp \rangle e_2e_2^*,
\end{align*}
giving rise to the following Markov transition matrix
$$P = \begin{bmatrix}
\langle b, \beta b \rangle & \langle b_\perp, \beta b_\perp \rangle \\
\langle c, \beta c \rangle & \langle c_\perp, \beta c_\perp \rangle
\end{bmatrix}.$$
Note that the matrix above has the general form of a $2 \times 2$ stochastic matrix.

If moreover \ref{it:row-initial} holds, the matrix $P$ is bistochastic.
\end{remark}

\section{The block--diagonal algebra}
\label{sec:block-diagonal}

Let 
$$\mathcal A=\bigoplus_{a=1}^N M_{d_a}(\mathbb C)$$
be a subalgebra of $M_n(\mathbb C)$ and 
\begin{equation}\label{eq:decomp-Cn-sum}
\mathbb C^n=\bigoplus_{a=1}^N V_a
\end{equation}
the associated direct sum of orthogonal subspaces decomposition. In the above expression, the integers $d_1, \ldots, d_N$ form a partition of $n$, while the subspaces $V_1, \ldots, V_N$ form a partition of $\mathbb C^n$. Our goal is to characterize UCP maps as in \eqref{eq:def-UCP} (resp.~TPCP as in \eqref{eq:def-TPCP}) which leave the algebra $\mathcal A$ invariant, irrespective of the ancilla state $\beta$. 
Let us consider a basis $\{e_1, \ldots e_n\}$ of $\mathbb C^n$ which is compatible with the decomposition \eqref{eq:decomp-Cn-sum}, that is there exists a partition $[n] = \sqcup_{a=1}^N {\bf a}$ such that $V_a = \operatorname{linspan} \{e_i\}_{i \in {\bf a}}$ for all $a \in [N]$. We define an equivalence relation $i\sim j\Leftrightarrow i,j\in {\bf a}$ for a unique $a\in[N]$.

We shall consider two different block decompositions of the unitary interaction matrix $U \in \mathcal U_{nk}$: the usual one
$$U = \sum_{i,j=1}^n e_ie_j^* \otimes \tilde U_{ij},$$
with $\tilde U_{ij} \in M_k(\mathbb C)$, and the decomposition corresponding to the coarser partition \eqref{eq:decomp-Cn-sum}
$$U=\begin{bmatrix}
U_{11}&U_{12}&\cdots&U_{1N}\\U_{21}&U_{22}&\ldots&U_{2N}\\\vdots&{}&\ddots&\vdots\\U_{N1}&\cdots&\cdots&U_{NN}
\end{bmatrix},$$
where
$$ U_{ab} = \sum_{\substack{i \in {\bf a} \\ j \in {\bf b}}} e_ie_j^* \otimes \tilde U_{ij} \in M_{d_ak \times d_bk}(\mathbb C).$$
Accordingly, let $\{E_{ab}\}_{a,b\in[N]}$ (resp. $\{F_{ab}\}_{a,b\in[N]}$) be the initial (resp. final) spaces of the matrices $ U_{ab}$: $E_{ab} = (\ker  U_{ab})^\perp$ and $F_{ab} = \operatorname{ran}  U_{ab}$.

\subsection{The Heisenberg picture}

As before, we start with the Heisenberg picture of UCP maps, since the statement and the proof of the main result are easier. Note that the Theorem below generalizes Theorem \ref{thm:diagonal-Heisenberg}, which corresponds to the particular case $d_a=1$ for all $a \in [N]$.

\begin{theorem}\label{thm:block-diagonal-Heisenberg}
Let $U\in\mathcal U_{nk}$ and $\mathcal A$ the above defined subalgebra of $M_n(\mathbb C)$. The following two statements are equivalent:
\begin{enumerate}
\item There exists a set of quantum states $\mathcal B$ which spans $M_k(\mathbb C)$ such that for any $\beta\in\mathcal B$, the subalgebra $\mathcal A$ is stable by the UCP map $S_{\beta,U}$:
$$\forall \beta\in\mathcal B,\qquad S_{U,\beta}(\mathcal A)\subset\mathcal A.$$
\item The matrices $\{U_{ab}\}_{a,b\in[N]}$ are partial isometries from $V_b\otimes \mathbb C^k$ to $V_a\otimes \mathbb C^k$ such that
\begin{enumerate}
\item For any $a,b\in[N]$, $F_{ab}=V_a\otimes \hat F_{ab}$ and for each $a\in [N]$, the family $\{\hat F_{ab}\}_{b\in[N]}$ is a partition of $\mathbb C^k$;
\item For each $b\in [N]$, the family $\{E_{ab}\}_{a\in[N]}$ is a partition of $V_b\otimes\mathbb C^k.$
\end{enumerate}
\end{enumerate}
\end{theorem}
\begin{proof}
Let us show $(1) \implies (2)$. We initially follow the same path as for Theorem \ref{thm:diagonal-Heisenberg}. By direct computation, it is easy to see that the hypothesis from (1) is equivalent to
\begin{equation}\forall i\sim j\text{ and }x\not\sim y, \qquad \tilde U_{ix}^* \tilde U_{jy}=0.\label{eq:detailed_cond}\end{equation}
It remains to prove the equivalence of (2) with this statement.

We can reformulate \eqref{eq:detailed_cond} in terms of matrices:
\begin{equation}\label{eq:detailed-cond-ab}
\forall a,b,c\in[N], b\neq c, \forall X\in M_{d_a}(\mathbb C),\quad U_{ab}^*(X\otimes I_k) U_{ac}=0.
\end{equation}
The proof of this equivalence amounts to taking specific matrices $X=e_i e_j^*$ and multiplying on the left and right by $e_xe_x^*\otimes I_k$ and $e_ye_y^*\otimes I_k$ respectively. We leave the details to the reader.

Let us prove that \eqref{eq:detailed-cond-ab} implies
\begin{enumerate}
\item For all $a\in[N]$, the family $\{F_{ab}\}_{b\in[N]}$ is a partition of $V_a\otimes\mathbb C^k$;
\item For all $b\in[N]$, the family $\{E_{ab}\}_{a\in[N]}$ is a partition of $V_b\otimes\mathbb C^k$.
\end{enumerate}
On the one hand, taking $X=I_{d_a}$ in \eqref{eq:detailed-cond-ab}, we deduce that for each $a\in[N]$ and $b\neq c$, $F_{ab}\perp F_{ac}$. It follows that $F_a=\bigoplus_b F_{ab}$ is a subspace of $V_a\otimes \mathbb C^k$ and is a direct sum of orthogonal subspaces thus $\dim F_a=\sum_{b\in[N]}\dim F_{ab}\leq d_ak$.

On the other hand, since $U$ is unitary, $\sum_{b}U_{ab}U_{ab}^*=I_{d_ak}$. Hence from Cochran's Lemma \ref{lem:Cochran}, $\sum_{b} \dim F_{ab}\geq d_ak$. Thus $\sum_{b} \dim F_{ab}= d_ak$ and the matrices $U_{ab}$ are partial isometries with $\{F_{ab}\}_b$ a partition of $V_a\otimes C^k$. Again, since $U$ is unitary $\sum_{a}U_{ab}^*U_{ab}=\sum_a P_{E_{ab}}=I_{d_bk}$. Hence $\{E_{ab}\}_{a\in[N]}$ is a partition of $V_b\otimes \mathbb C^k$.

We now show that \eqref{eq:detailed-cond-ab} implies a finer structure for the final spaces.

Let $a,b,c\in[N]$ such that $b\neq c$. Let $\phi\in F_{ab}$ and $\psi\in F_{ac}$. Assuming \eqref{eq:detailed-cond-ab}, we have that for all $X\in M_{d_a}(\mathbb C)$,  $\langle\phi,(X\otimes I_k)\psi\rangle=0$. Let the families $\{\phi_i\}_{i\in {\bf a}}\subset \mathbb C^k$ and $\{\psi_i\}_{i \in {\bf a}} \subset \mathbb C^k$ be such that $\phi=\sum_{i\in {\bf a}} e_i\otimes \phi_i$ and $\psi=\sum_{i\in {\bf a}} e_i\otimes \psi_i$. Set $X\in M_{d_a}(\mathbb C)$ such that  in the basis $\{e_i\}_{i\in {\bf a}}$,
$$X_{ij}=\left\{\begin{array}{ll} \frac{|\langle \phi_i,\psi_j\rangle|}{\langle \phi_i,\psi_j\rangle}&\text{ if }\langle \phi_i,\psi_j\rangle\neq 0\\ 1 & \text{ if }\langle \phi_i,\psi_j\rangle= 0\end{array}\right.$$
Then $\langle\phi,(X\otimes I_k)\psi\rangle=0$ implies
$$\sum_{i,j\in{\bf a}}|\langle \phi_i,\psi_j\rangle|=0.$$
Hence for all $i,j\in{\bf a}, \phi_i\perp\psi_j$.

Let $\hat F_{ab}=\operatorname{linspan}\{\phi_i|\phi\in F_{ab}, \phi=\sum_{i\in {\bf a}}e_i\otimes \phi_i\}$ and $\hat F_{ac}=\operatorname{linspan}\{\psi_i|\psi\in F_{ac}, \psi=\sum_{i\in {\bf a}}e_i\otimes \psi_i\}$. We have $\hat F_{ab}\perp \hat F_{ac}$ and $F_{ab}\subset V_a\otimes \hat F_{ab}$ for all $b\in[N]$. We shall prove the converse inclusion.

Let $x\in V_a$ and $\phi\in \hat F_{ab}$. Then $x\otimes \phi\in V_a\otimes \hat F_{ab}$ hence $\forall\psi\in F_{ac}, \langle x\otimes \phi,\psi\rangle=0$ for any $c\neq b$. Since $\{F_{ab}\}_{b\in[N]}$ is a partition of $V_a\otimes \mathbb C^k$, $x\otimes\phi\in F_{ab}$. Hence $V_a\otimes \hat F_{ab}\subset F_{ab}$ and $F_{ab}=V_a\otimes \hat F_{ab}$. Since $\{F_{ab}\}_{b\in[N]}$ is a partition of $V_a\otimes \mathbb C^k$, $\{\hat F_{ab}\}_{b\in[N]}$ is a partition of $\mathbb C^k$.

For the converse implication $(2) \implies (1)$, assume (a) and let us prove \eqref{eq:detailed-cond-ab}. Let $(\phi,\psi)\in V_c\otimes \mathbb C^k \times V_b\otimes \mathbb C^k$, for some blocks $b \neq c$. Then $U_{ab}\psi\in V_a\otimes \hat F_{ab}$ and $U_{ac}\phi\in V_a\otimes \hat F_{ac}$. Then for any $X\in M_{d_a}(\mathbb C)$, 
\begin{equation*}
\begin{split}
\langle \psi, U_{ab}^*(X\otimes I_k)U_{ac}\phi\rangle&=\langle U_{ab}\psi, (I_{d_a}\otimes P_{\hat F_{ab}})(X\otimes I_k)(I_{d_a}\otimes P_{\hat F_{ac}})U_{ac}\phi\rangle\\
	&=\langle U_{ab}\psi,(X\otimes P_{\hat F_{ab}}P_{\hat F_{ac}})U_{ac}\phi\rangle\\
	&=0.
\end{split}
\end{equation*}
We deduce \eqref{eq:detailed-cond-ab} and thus (1).
\end{proof}

At the level of examples, the identity and the flip operators from Example \ref{ex:id-flip} satisfy also the conditions in Theorem \ref{thm:block-diagonal-Heisenberg}.

\subsection{The Schr\"odinger picture}

We move now to the Schr\"odinger case. Since the proof follows closely what has been done in the previous sections, we only sketch the main ideas, leaving the details to the reader. 

\begin{theorem}\label{thm:block-diagonal-Schrodinger}
Let $U\in\mathcal U_{nk}$ and $\mathcal A$ the above defined subalgebra of $M_n(\mathbb C)$. The following two statements are equivalent:
\begin{enumerate}
\item There exists a set of quantum states $\mathcal B$ which spans $M_k(\mathbb C)$ such that for any $\beta\in\mathcal B$, the subalgebra $\mathcal A$ is stable by the quantum channel $T_{\beta,U}$:
$$\forall \beta\in\mathcal B,\qquad T_{U,\beta}(\mathcal A)\subset\mathcal A.$$
\item The matrices $\{U_{ab}\}_{a,b\in [N]}$ are partial isometries from $V_b\otimes \mathbb C^k$ to $V_a\otimes \mathbb C^k$ such that,
\begin{enumerate}
\item For any $a,b\in[N]$, $F_{ab}=V_a\otimes \hat F_{ab}$ and for each $a\in[N]$, the family $\{\hat F_{ab}\}_{b\in[N]}$ is a partition of $\mathbb C^k$;
\item For each $b\in[N]$, the family $\{F_{ab}\}_{a\in[N]}$ is a partition of $V_b\otimes \mathbb C^k$;
\item For each $b\in[N]$, the family $\{E_{ab}\}_{a\in[N]}$ is a partition of $V_b\otimes \mathbb C^k$.
\end{enumerate}
\end{enumerate}
\end{theorem}
\begin{proof}
From direct computation, (1) is equivalent to
\begin{equation}\label{eq:bloc-diagonal-H-condition-small-blocs}
\forall i\sim j\text{ and } x\nsim y,\quad\tilde{U}_{yj}^\ast \tilde{U}_{xi}=0.
\end{equation}
The implication $(2)\implies (1)$ proof is then identical to the one of Theorem \ref{thm:block-diagonal-Heisenberg}. We thus only show the implication, $(1) \implies (2)$. 

In terms of matrices, \eqref{eq:bloc-diagonal-H-condition-small-blocs} is equivalent to,
\begin{equation}\label{eq:bloc-diagonal-condition-schrodinger-X}
\forall a,b,c\in[N], b\neq c, \forall X\in M_{d_b\times d_c}, \quad U_{ba}^\ast (X\otimes I_k)U_{ca}=0.
\end{equation}
Let us prove that,
\begin{enumerate}
\item \label{it:blocs-schrodinger-final-spaces-on-second-index} For each $a\in[N]$, the family $\{F_{ab}\}_{b\in[N]}$ is a partition of $V_a\otimes \mathbb C^k$;
\item \label{it:blocs-schrodinger-final-spaces-on-first-index} For each $b\in[N]$, the family $\{F_{ab}\}_{a\in[N]}$ is a partition of $V_b\otimes \mathbb C^k$;
\item \label{it:blocs-schrodinger-initial-spaces-on-first-index} For each $b\in[N]$, the family $\{E_{ab}\}_{a\in[N]}$ is a partition of $V_b\otimes \mathbb C^k$.
\end{enumerate}
Fix $a,b,c \in [N]$, $b\neq c$ such that $d_b\geq d_c$. Let,
$$X=U\left(\begin{array}{c}I_{d_c}\\0\end{array}\right)$$
with $U$ a unitary $d_b\times d_b$ matrix. Spanning over all possible $U$, equation \eqref{eq:bloc-diagonal-condition-schrodinger-X} implies $F_{ba}\perp F_{ca}$. It follows that for every $a\in [N]$, $\sum_{b\in [N]}\operatorname{rk}(U_{ba})\leq d_a\,k$. Since $\sum_{b\in[N]} U_{ba}^*U_{ba}=I_{d_a}\otimes I_k$, from Cochran Lemma \ref{lem:Cochran}, $\sum_{b\in[N]} \operatorname{rk}(U_{ba})\geq k$ and for any $a,b\in[N]$, $U_{ab}$ is a partial isometry such that  for any $a\in[N]$ the family $\{F_{ba}\}_{b\in[N]}$ is a partition of $V_a\otimes \mathbb C^k$ and (\ref{it:blocs-schrodinger-final-spaces-on-second-index}) holds.

Since for any $b\in[N]$, $\sum_{a\in[N]} U_{ba}U_{ba}^*=I_{d_b}\otimes I_k$ and  for any $a\in[N]$, $\sum_{b\in[N]} U_{ba}^*U_{ba}=I_{d_a}\otimes I_k$, the fact that the matrices $U_{ab}$ are partial isometries implies (\ref{it:blocs-schrodinger-final-spaces-on-first-index}) and (\ref{it:blocs-schrodinger-initial-spaces-on-first-index}) hold.

We now show that \eqref{eq:bloc-diagonal-condition-schrodinger-X} implies a finer structure for the final spaces.

Let $a,b,c\in[N]$ such that $b\neq c$. Let $\phi\in F_{ba}$ and $\psi\in F_{ca}$. Assuming \eqref{eq:bloc-diagonal-condition-schrodinger-X}, we have that for all $X\in M_{d_b\times d_c}(\mathbb C)$,  $\langle\phi,(X\otimes I_k)\psi\rangle=0$. Let the families $\{\phi_i\}_{i\in {\bf b}}\subset \mathbb C^k$ and $\{\psi_i\}_{i \in {\bf c}} \subset \mathbb C^k$ be such that $\phi=\sum_{i\in {\bf b}} e_i\otimes \phi_i$ and $\psi=\sum_{i\in {\bf c}} e_i\otimes \psi_i$. Set $X\in M_{d_b\times d_c}(\mathbb C)$ such that  in the orthonormal bases $\{e_i\}_{i\in {\bf b}}$ and $\{e_j\}_{j\in {\bf c}}$ of respectively $V_b$ and $V_c$,
$$X_{ij}=\left\{\begin{array}{ll} \frac{|\langle \phi_i,\psi_j\rangle|}{\langle \phi_i,\psi_j\rangle}&\text{ if }\langle \phi_i,\psi_j\rangle\neq 0\\ 1 & \text{ if }\langle \phi_i,\psi_j\rangle= 0\end{array}\right.$$
Then $\langle\phi,(X\otimes I_k)\psi\rangle=0$ implies
$$\sum_{(i,j)\in {\bf b}\times {\bf c}}|\langle \phi_i,\psi_j\rangle|=0.$$
Hence for all $(i,j)\in{\bf b}\times{\bf c}, \phi_i\perp\psi_j$.

Let $\hat F_{ba}=\operatorname{linspan}\{\phi_i|\phi\in F_{ba}, \phi=\sum_{i\in {\bf b}}e_i\otimes \phi_i\}$ and $\hat F_{ca}=\operatorname{linspan}\{\psi_i|\psi\in F_{ca}, \psi=\sum_{i\in {\bf c}}e_i\otimes \psi_i\}$. We have $\hat F_{ba}\perp \hat F_{ca}$ and $F_{ba}\subset V_b\otimes \hat F_{ba}$ for all $b\in[N]$. We shall prove the converse inclusion.

Let $x\in V_b$ and $\phi\in \hat F_{ba}$. Then $x\otimes \phi\in V_b\otimes \hat F_{ba}$ hence $\forall\psi\in F_{ca}, \langle x\otimes \phi,\psi\rangle=0$ for any $c\neq b$. Since $\{F_{ba}\}_{b\in[N]}$ is a partition of $V_a\otimes \mathbb C^k$, $x\otimes\phi\in F_{ba}$. Hence $V_b\otimes \hat F_{ba}\subset F_{ba}$ and $F_{ba}=V_b\otimes \hat F_{ba}$. Since $\{F_{ba}\}_{a\in[N]}$ is a partition of $V_b\otimes \mathbb C^k$, $\{\hat F_{ba}\}_{a\in[N]}$ is a partition of $\mathbb C^k$.
\end{proof}

\subsection{Comparing subalgebras}
\label{sec:comparing-sum}

In this section, we would like to investigate the relation between the sets
$$\mathcal U^{(\text{inv})}_i = \{U \in \mathcal U_{nk} \, : \, S_{U,\beta}(\mathcal A_i) \subseteq \mathcal A_i \},$$
for the subalgebras $\mathcal A_{1,2} \subseteq M_n(\mathbb C)$ with the property that $\mathcal A_1 \subseteq \mathcal A_2$, focusing on the Heisenberg picture. Note that the two trivial extremal cases, $\mathcal A = \mathbb C I$ and $\mathcal A = M_n(\mathbb C)$ lead both to the set of all unitary operators. In a situation where both algebras are non-trivial, we claim that the two sets of unitary operators leaving them invariant are not comparable. We shall consider the simplest example, and leave the extension to the general case to the reader. 

Let $n=3$ and consider
\begin{align*}
\mathcal A_1 &= M_1(\mathbb C) \oplus M_1(\mathbb C) \oplus M_1(\mathbb C) \subset M_3(\mathbb C)\\
\mathcal A_2 &= M_2(\mathbb C) \oplus M_1(\mathbb C) \subset M_3(\mathbb C)
\end{align*}

To show that $\mathcal U^{(\text{inv})}_1 \nsubseteq \mathcal U^{(\text{inv})}_2$, consider the following operator
$$U = \begin{bmatrix}
e_1 e_1^* & e_2 e_1^* & e_3 e_1^* \\
e_3 e_2^* & e_1 e_2^* & e_2 e_2^* \\
e_2 e_3^* & e_3 e_3^* & e_1 e_3^* 
\end{bmatrix}.$$

Obviously, the operator $U$ above satisfies the conditions \ref{it:row-final}, \ref{it:col-initial} from Definition \ref{def:bistochastic-PI}, so $U \in \mathcal U^{(\text{inv})}_1$. However, for the equivalence relation $1 \sim 2 \nsim 3$ induced by the subalgebra $\mathcal A_2$, and the choice of indices $(i,j,x,y) = (1,2,2,3)$, equation \eqref{eq:detailed_cond} from the proof of Theorem \ref{thm:block-diagonal-Heisenberg} is not satisfied:
$$\tilde U_{ix}^* \tilde U_{jy} = (e_2e_1^*)^* e_2 e_2^* = e_1 e_2^* \neq 0,$$
hence $U \notin \mathcal U^{(\text{inv})}_2$.

To show that the reversed inclusion also fails to hold, consider a unitary matrix $W \in \mathcal U_4$ partitioned in $2 \times 2$ blocks $W_{ij}$: 
$$W = \begin{bmatrix}
W_{11} & W_{12} \\
W_{21} & W_{22}
\end{bmatrix}.$$
We construct the following unitary matrix $U \in \mathcal U_9$:
$$U = \begin{bmatrix}
0 & 0 & 0 & 0 & 0 & 0 & 0 & 1 & 0 \\
0 & \multicolumn{2}{c}{\multirow{2}{*}{$W_{11}$}} & 0 & \multicolumn{2}{c}{\multirow{2}{*}{$W_{12}$}} & 0 & 0 & 0 \\
0 & & & 0 &  &  & 0 & 0 & 0 \\
0 & 0 & 0 & 0 & 0 & 0 & 0 & 0 & 1 \\
0 & \multicolumn{2}{c}{\multirow{2}{*}{$W_{21}$}} & 0 & \multicolumn{2}{c}{\multirow{2}{*}{$W_{22}$}} & 0 & 0 & 0 \\
0 &  &  & 0 &  &  & 0 & 0 & 0 \\
0 & 0 & 0 & 0 & 0 & 0 & 1 & 0 & 0 \\
1 & 0 & 0 & 0 & 0 & 0 & 0 & 0 & 0 \\
0 & 0 & 0 & 1 & 0 & 0 & 0 & 0 & 0
\end{bmatrix}.$$

One can check by direct computation that $U \in \mathcal U^{(\text{inv})}_2$. However if one chooses $W$ in such a way that one of the blocks $W_{ij}$ is \emph{not} a partial isometry, then $U \notin \mathcal U^{(\text{inv})}_1$. This is always possible, since one can simply choose any small enough matrix $W_{11} \in M_2(\mathbb C)$ and complete it to a full unitary matrix $W \in \mathcal U_4$.

\section{The tensor product algebra}\label{sec:tensor-product}

We now analyze the case where there is just one term in the direct sum in \eqref{eq:C-star-subalgebra}, that is 
$$\mathcal A = M_d(\mathbb C) \otimes I_r.$$
In particular, the system Hilbert space decomposes as $\mathbb C^n = \mathbb C^d \otimes \mathbb C^r$, with $n=dr$. We shall denote by $\omega_k \in \mathbb C^k \otimes \mathbb C^k$ the (un-normalized) \emph{maximally entangled state}
\begin{equation}\label{eq:def-omega}
\omega_k = \sum_{i=1}^k e_i \otimes e_i,
\end{equation}
for some fixed orthonormal basis $\{e_1, \ldots, e_k\}$ of $\mathbb C^k$.

\subsection{The Heisenberg picture}

\begin{theorem}\label{thm:tensor-product-H}
Let $U \in \mathcal U_{nk}$ be a bipartite unitary operator. The following are equivalent:
\begin{enumerate}
\item \label{it:tensor-product-H-i} There exists a set of quantum states $\mathcal B$ which spans $M_k(\mathbb C)$ such that for any $\beta\in\mathcal B$, the subalgebra $\mathcal A$ is stable by the UCP map $S_{\beta,U}$:
$$\forall \beta\in\mathcal B,\qquad S_{U,\beta}(\mathcal A)\subset\mathcal A.$$
\item \label{it:tensor-product-H-ii} There exist unitary operators $V \in \mathcal U_{rk}$ and $W \in \mathcal U_{dk}$ such that
\begin{equation}\label{eq:U-tensor-product-H}
U = (I_d \otimes V) \cdot (W \otimes I_r).
\end{equation}
\end{enumerate}
\end{theorem}
\begin{proof}
We first show $\eqref{it:tensor-product-S-ii}  \implies \eqref{it:tensor-product-S-i}$. For an interaction unitary operator $U$ as in \eqref{eq:U-tensor-product-H}, we have
\begin{align*}
S_{U, \beta}(X \otimes I_r) &= (\operatorname{id}_d \otimes \operatorname{id}_r \otimes \operatorname{Tr}_k) \left[ (W^* \otimes I_r)  (I_d \otimes V^*)  (X \otimes I_r \otimes I_k) (I_d \otimes V) (W \otimes I_r) (I_{dr} \otimes \beta) \right] \\
 &= (\operatorname{id}_d \otimes \operatorname{id}_r \otimes \operatorname{Tr}_k) \left[ (W^* \otimes I_r)   (X \otimes I_r \otimes I_k)  (W \otimes I_r) (I_{dr} \otimes \beta) \right] \\
  &= (\operatorname{id}_d  \otimes \operatorname{Tr}_k) \left[ W^*   (X  \otimes I_k)  W (I_{d} \otimes \beta) \right] \otimes I_r \\
&= S_{W,\beta}(X) \otimes I_r \in \mathcal A,
\end{align*}
proving the claim; a graphical representation of the computation above can be found in Figure \ref{fig:U-tensor-product-H}.

\begin{figure}[htbp] 
\includegraphics{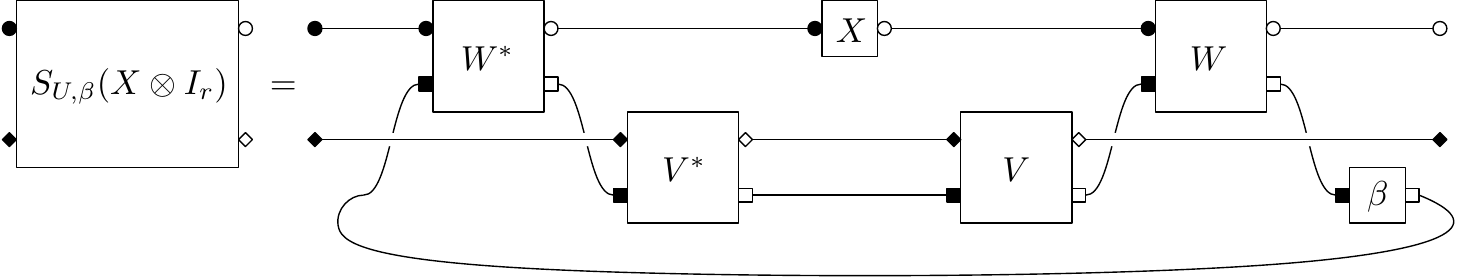}\\\vspace{.6cm}
\includegraphics{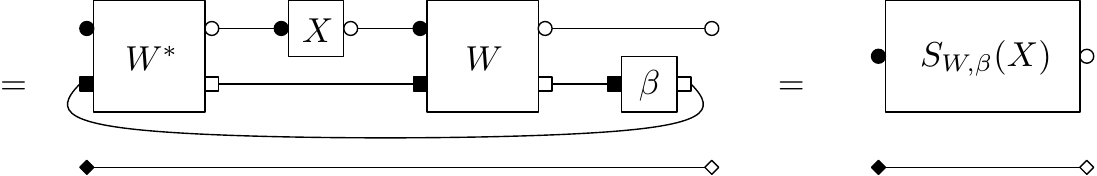}
\caption{Bipartite unitary operators as in \eqref{eq:U-tensor-product-H} yield $\mathcal A$-preserving UCP maps. Round, diamond, and respectively square decorations correspond to the Hilbert spaces $\mathbb C^d$, $\mathbb C^r$, and respectively $\mathbb C^k$.} 
\label{fig:U-tensor-product-H}
\end{figure}

For the reverse implication, using linearity and the hypothesis, we find that there is a bi-linear function $f_U(X, \beta)$ such that
$$\forall X \in M_n(\mathbb C), \, \beta \in M_k(\mathbb C), \qquad S_{U,\beta}(X \otimes I_r) = f_U(X, \beta) \otimes I_r.$$
We rewrite the last equation as 
$$(\operatorname{id}_d \otimes \operatorname{id}_r \otimes \operatorname{Tr}_k) \left[ (U^\Gamma)^*  (X \otimes I_r \otimes \beta^\top) U^\Gamma \right] = f_U(X, \beta) \otimes I_r,$$
where $U^\Gamma$ is the partial transpose of the unitary operator $U$ with respect to the $\mathbb C^k$ tensor factor: $U^\Gamma = (\operatorname{id}_d \otimes \operatorname{id}_r \otimes \operatorname{transp}_k)(U)$.
From the universal property of the tensor product $M_n(\mathbb C) \otimes M_k(\mathbb C)$, we find that there exists a linear function $g_U: M_{dk}(\mathbb C) \to M_d(\mathbb C)$ such that
$$(\operatorname{id}_d \otimes \operatorname{id}_r \otimes \operatorname{Tr}_k) \left[ (U^\Gamma)^*  (X \otimes I_r \otimes \beta^\top) U^\Gamma \right] = g_U(X \otimes \beta^\top) \otimes I_r.$$
Moreover, since the left hand side above is a completely positive function of the input $X \otimes \beta^\top$, the same must hold for $g_U$. Hence, from the Kraus form of CP maps, we can find, for some $s \geq 1$, a linear operator $A: \mathbb C^{ds} \to \mathbb C^{dk}$ such that $g_U(X \otimes \beta) = (\operatorname{id}_d \otimes \operatorname{Tr}_s)\left[ A^* (X \otimes \beta) A \right]$; note that we have dropped the transposition operation on $\beta$, using the fact that a family $\mathcal B$ spans $M_k(\mathbb C)$ iff $\mathcal B^\top$ does. We have thus, for all $X$ and $\beta$, 
$$(\operatorname{id}_d \otimes \operatorname{id}_r \otimes \operatorname{Tr}_k) \left[ (U^\Gamma)^*  (X \otimes I_r \otimes \beta) U^\Gamma \right] = (\operatorname{id}_d \otimes \operatorname{id}_r \otimes \operatorname{Tr}_s) \left[ A^*  (X \otimes I_r \otimes \beta) A \right],$$
where we recognize two different dilations of a CP map. On the level of the Choi matrices, the equality above translates to 
\begin{align*}
(\operatorname{id}_d \otimes \operatorname{id}_d \otimes \operatorname{id}_r \otimes \operatorname{Tr}_k \otimes \operatorname{id}_k) \left[(I_d \otimes (U^\Gamma)^* \otimes I_k ) (\omega_d\omega_d^* \otimes I_r \otimes \omega_k\omega_k^*) (I_d \otimes U^\Gamma \otimes I_k)  \right] = \\
(\operatorname{id}_d \otimes \operatorname{id}_d \otimes \operatorname{Tr}_s \otimes \operatorname{id}_k) \left[(I_d \otimes A^* \otimes I_k ) (\omega_d\omega_d^* \otimes \omega_k\omega_k^*) (I_d \otimes A \otimes I_k)  \right] \otimes I_r \iff \\
(I_d \otimes U^* ) (\omega_d\omega_d^* \otimes I_r \otimes I_k) (I_d \otimes U) = \left[(I_d \otimes (A^\Gamma)^*) (\omega_d\omega_d^* \otimes I_s) (I_d \otimes A^\Gamma )\right]\otimes I_r,
\end{align*}
where $A^\Gamma : \mathbb C^d \otimes \mathbb C^k \to \mathbb C^d \otimes \mathbb C^s$ is the partial transposition of the operator $A$ with respect to the second subsystems. The two square roots of the last equation above are related by a partial isometry  $V:\mathbb C^{rs} \to \mathbb C^{rk}$ as follows
\begin{align*}
(\omega_d^* \otimes I_r \otimes I_k)(I_d \otimes U) = (I_d \otimes I_d \otimes V)\left[(\omega_d \otimes I_s)(I_d \otimes A^\Gamma) \otimes I_r \right] \iff
U = (I_d \otimes V)(A^\Gamma \otimes I_r).
\end{align*}
Let us set $W:=A^\Gamma$. Since $U$ is an unitary operator, we have 
$$I_{drk} = U^*U = (W^* \otimes I_r) (I_d \otimes V^*) (I_d \otimes V) (W \otimes I_r) .$$
and thus $V^*V$ must have full rank, i.e.~$V$ is an isometry; in particular, $s \leq k$. But then, $I_{drk} = W^*W \otimes I_r$, hence $W$ is also an isometry, implying $k \leq s$ . We have thus $s=k$, and $V,W$ are unitary operators, as claimed. 
\end{proof}

Let us consider some special cases of the result above. In the $d=1$ or $r=1$ case, requiring that the algebra $\mathcal A$ should be preserved does not impose additional constraints, so we recover the set of all unitary matrices $\mathcal U_{nk}$. If $k=1$, the ancilla space is trivial, and the UCP map acts by unitary conjugation. It is clear then that the unitary operator $U$ must be a tensor product, $U=W \otimes V$, with $W \in \mathcal U_d$ and $V \in \mathcal U_r$.

\subsection{The Schr\"odinger picture}

Before stating our main result, we introduce the set of bipartite unitary operators with the property that their partial transpose is also unitary
\begin{equation}\label{eq:Uunital}
\mathcal U_{unital}:= \{U \in \mathcal U_{nk} \, : \, U^\Gamma \in \mathcal U_{nk} \} = \mathcal U_{nk} \cap \mathcal U_{nk}^\Gamma.
\end{equation}

This set has been studied in \cite{dnp}, and we refer the reader to that paper for additional properties of such unitary operators.

\begin{theorem}\label{thm:tensor-product-S}
Let $U \in \mathcal U_{nk}$ be a bipartite unitary operator. The following are equivalent:
\begin{enumerate}
\item \label{it:tensor-product-S-i} There exists a set of quantum states $\mathcal B$ which spans $M_k(\mathbb C)$ such that for any $\beta\in\mathcal B$, the subalgebra $\mathcal A$ is stable by the quantum channel $T_{\beta,U}$:
$$\forall \beta\in\mathcal B,\qquad T_{U,\beta}(\mathcal A)\subset\mathcal A.$$
\item \label{it:tensor-product-S-ii} There exist unitary operators $V \in \mathcal U_{dk}$ and $W \in \mathcal U_{rk} \cap \mathcal U_{rk}^\Gamma$ such that
\begin{equation}\label{eq:U-tensor-product-S}
U = (I_d \otimes W^\Gamma) \cdot (V \otimes I_r).
\end{equation}
\end{enumerate}
\end{theorem}

\begin{proof}
Let us first show $\eqref{it:tensor-product-S-ii}  \implies \eqref{it:tensor-product-S-i}$. Consider a bipartite unitary $U$ as in \eqref{eq:U-tensor-product-S}; for any $A \in M_d(\mathbb C)$ and $\beta \in M_k(\mathbb C)$, we have
\begin{align*}
T_{U, \beta}(A \otimes I_r) &= (\operatorname{id}_d \otimes \operatorname{id}_r \otimes \operatorname{Tr}_k)\left[ (I_d \otimes W)  (V \otimes I_r)  (A \otimes I_r \otimes \beta)  (V^* \otimes I_r)  (I_d \otimes W^*) \right] \\
&= (I_d \otimes I_r \otimes \omega_k^* )  (V \otimes W^\Gamma)  (A \otimes \beta \otimes I_r \otimes I_r)  (V^* \otimes W^{\Gamma*})  (I_d \otimes I_r \otimes \omega_k)\\
&= (\operatorname{id}_d \otimes \operatorname{Tr}_k)\left[ V  (A \otimes \beta)  V^* \right] \otimes I_r \in \mathcal A,
\end{align*}
where $\omega$ is the maximally entangled state \eqref{eq:def-omega}. The computation above is represented graphically in Figure \ref{fig:U-tensor-product-S}.

\begin{figure}[htbp] 
\includegraphics{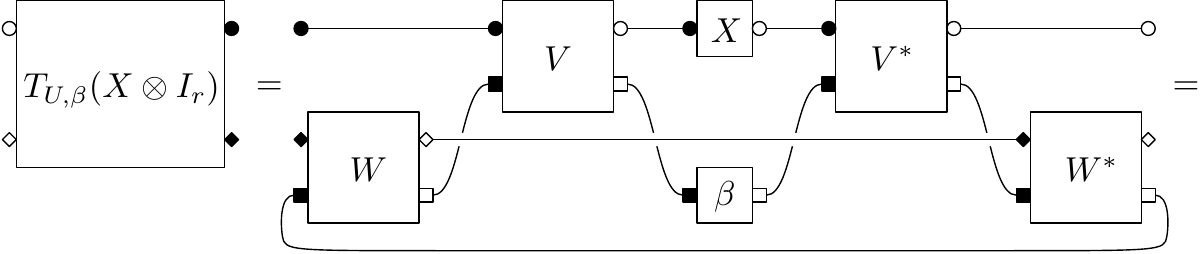}\\\vspace{.6cm}
\includegraphics{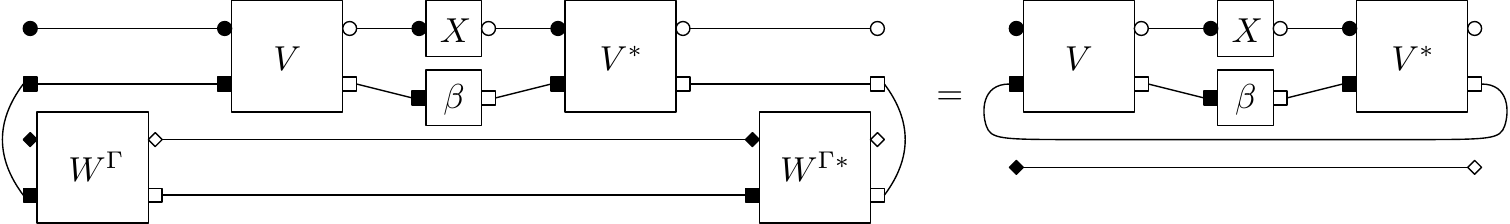}
\caption{Bipartite unitary operators as in \eqref{eq:U-tensor-product-S} yield $\mathcal A$-preserving channels.} 
\label{fig:U-tensor-product-S}
\end{figure}

Let us now prove $\eqref{it:tensor-product-S-i}  \implies \eqref{it:tensor-product-S-ii}$. We follows the steps in the proof of Theorem \ref{thm:tensor-product-H}. Using linearity, the universal property of the tensor product and complete positivity, we find that there are operators $V:\mathbb C^{dk} \to \mathbb C^{ds}$ and $W:\mathbb C^{kr} \to \mathbb C^{sr}$, for some $s \geq 1$, such that 
$$\omega_k^* U = \omega_s (V \otimes W),$$
see Figure \ref{fig:UVW-tensor-product-S}, left panel. Moreover, the application $W^*$ is an isometry, so $s \leq k$. Denoting by $W^\Gamma$ the partial transposition of $W$ with respect to the spaces $\mathbb C^k$ and $\mathbb C^s$, that is
$$W^\Gamma = \sum_{a,b=1}^r \sum_{i=1}^k \sum_{j=1}^s \langle e_i \otimes e_a ,W e_j \otimes e_b \rangle \, e_j \otimes e_a \cdot  e_i^* \otimes e_b^*,$$
we have
$$U=(I_d \otimes W^\Gamma)(V \otimes I_r),$$
as claimed; see Figure \ref{fig:UVW-tensor-product-S}, right panel. From the trace preservation condition, we get that $V$ is an isometry, and thus $k\leq s$. We conclude that $k=s$, and thus both $V$ and $W$ are unitary operators. Asking that $UU^*=I_{dkr}$ yields $G^\Gamma$ unitary, proving the final claim and finishing the proof. 

\begin{figure}[htbp] 
\includegraphics{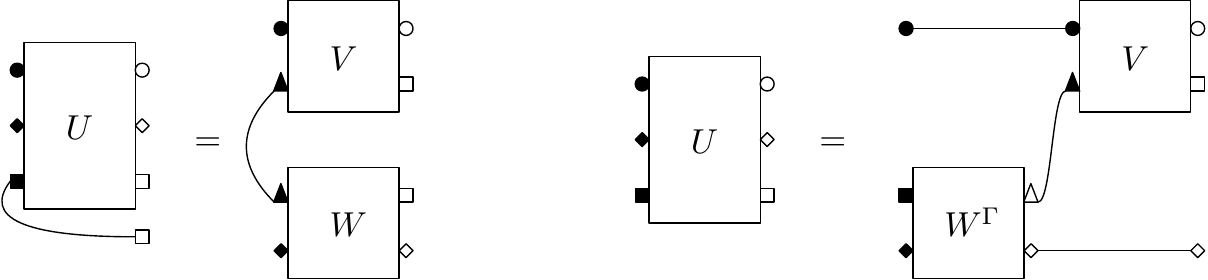}
\caption{The relation between the operators $U,V,W$. Triangle shaped labels correspond to the space $\mathbb C^s$; in the proof, it is shown that $s=k$.} 
\label{fig:UVW-tensor-product-S}
\end{figure}
\end{proof}

As in the Heisenberg case, we discuss next some special cases of Theorem \ref{thm:tensor-product-S}. In the case where $r=1$, $\mathcal A$ is the full input matrix algebra, so we recover the set of all unitary matrices $\mathcal U_{nk}$. If $k=1$, the ancilla space is trivial, and the quantum channel acts by unitary conjugation; hence, $U$ must be a tensor product, $U=V \otimes W$, with $V \in \mathcal U_d$ and $W \in \mathcal U_r$. Finally, if $d=1$, we just require of the quantum channels $T_{U,\beta}$ to be unital, independently of the value of $\beta$; but this is precisely Theorem \cite[Theorem 3.1]{dnp}: the unitary operator $U$ must be such that its partial transpose is also unitary $U = W \in \mathcal U_{rk} \cap \mathcal U_{rk}^\Gamma$.

\subsection{Comparing subalgebras}
\label{sec:comparing-tensor}

We would like to address here the same question as the one in Section \ref{sec:comparing-sum}, but with the following choice of subalgebras:
\begin{align*}
\mathcal A_1 &= M_d(\mathbb C) \otimes I_r\\
\mathcal A_2 &= \underbrace{M_d(\mathbb C) \oplus M_d(\mathbb C) \oplus \cdots \oplus M_d(\mathbb C)}_{r \text{ times}},
\end{align*}
where $d,r \geq 2$ are two arbitrary integers, and $n=dr$. It is clear that $\mathcal A_1 \subseteq \mathcal A_2$. As in the previous case, we show that the two sets of unitary operators that leave invariant the above subalgebras cannot be compared. 

Let us start by considering a general element from $\mathcal U^{(\text{inv})}_1$, of the form $U = (I_d \otimes V)(W \otimes I_r)$, see Theorem \ref{thm:tensor-product-H}. In order to check the conditions in Theorem \ref{thm:block-diagonal-Heisenberg}, we decompose 
$$V = \sum_{a,b=1}^r V_{ab} \otimes e_a e_b^*,$$
and obtain the following expression for the blocks of $U$:
$$U_{ab} = (I_d \otimes V_{ab})W \in M_{dk}(\mathbb C).$$
Since $W$ is unitary, an operator as above is a partial isometry iff $V_{ab}$ is one; but one can easily construct a unitary matrix $V$ having at least one block which is not a partial isometry (one can pick any small enough matrix $V_{11}$ which is not a partial isometry and extend it to a unitary matrix). This shows that $\mathcal U^{(\text{inv})}_1 \nsubseteq \mathcal U^{(\text{inv})}_2$.

To show that the reversed inclusion also fails to hold, consider the case $k=r$ and choose two families $\{e_{ab}\}$, $\{f_{ab}\}$ of $k^2$ vectors, such that the conditions corresponding to conditions \ref{it:row-final}, \ref{it:col-initial} of Definition \ref{def:bistochastic-PI} are satisfied:
\begin{align*}
\forall b, \, &\{e_{ab}\}_a \text{ is an orthonormal basis of $\mathbb C^k$}\\
\forall a, \, &\{f_{ab}\}_b \text{ is an orthonormal basis of $\mathbb C^k$}.
\end{align*}
Define
$$U:= \sum_{a,b=1}^{r=k} g_ag_b^* \otimes \tilde V_{ab} \otimes f_{ab}e_{ab}^*,$$
where $\{g_a\}$ is another orthonormal basis of $\mathbb C^r = \mathbb C^k$ and $\{\tilde V_{ab}\}$ is an arbitrary family of unitary operators acting on $\mathbb C^d$. The operator defined above leaves invariant the algebra $\mathcal A_2$, since its blocks
$$U_{ab} = \tilde V_{ab} \otimes f_{ab}e_{ab}^*$$
are partial isometries satisfying the hypotheses of Theorem \ref{thm:block-diagonal-Heisenberg} with initial and final spaces given by
$$E_{ab} = \mathbb C^d \otimes \mathbb C e_{ab} \quad \text{and} \quad F_{ab} = \mathbb C^d \otimes \mathbb C f_{ab}.$$
Let us assume now that the unitary $U$ described above also leaves the algebra $\mathcal A_2$ invariant; from Theorem \ref{thm:tensor-product-H}, we find unitary operators $V \in \mathcal U_{k^2}$ and $W \in \mathcal U_{kd}$ such that $U = (I_d \otimes V)(W \otimes I_k)$. In particular, the blocks of $U$ read 
$$U_{ab} = (I_d \otimes V_{ab})W,$$
where $V_{ab}$ are the $k \times k$ blocks of the unitary operator $V$. Muliplying the equation above with the adjoint of the same expression for another index pair $(\alpha,\beta)$, we get
$$I_d \otimes V_{ab}V_{\alpha\beta}^* = \langle e_{ab} , e_{\alpha \beta} \rangle \tilde V_{ab} \tilde V_{\alpha \beta}^* \otimes f_{ab}f_{\alpha \beta}^*.$$
A contradiction occurs in the relation above if there is a quadruple $(a,b,\alpha, \beta)$ such that $e_{ab} \not\perp e_{\alpha \beta}$ and $\tilde V_{ab} \tilde V_{\alpha \beta}^* \notin \mathbb C I_d.$

\section{Zero block algebra}
\label{sec:zero_bloc}

Let
\[\mathcal A=0_{d_0}\oplus M_{d_1}(\mathbb C)\]
be a subalgebra of $M_n(\mathbb C)$ and
\[\mathbb C^n=V_0\oplus V_1\]
the associated direct sum of orthogonal subspaces decomposition. In this decomposition, $d_0+d_1=n$ and $\{V_0,V_1\}$ is a partition of $\mathbb C^n$. In this section we want to characterize the bipartite unitary matrices generating UCP maps as in \eqref{eq:def-UCP} (resp. TPCP maps as in \eqref{eq:def-TPCP}) preserving $\mathcal A$ irrespective of the ancilla state $\beta$. It turns out that imposing the stability of $\mathcal A$ for a given positive definite $\beta$ is enough.

We consider the block decomposition of the bipartite matrices corresponding to the partition of $\mathbb C^n$ associated to $\mathcal A$. Namely for any $X\in M_{nk}(\mathbb C)\cong M_{n}(M_k(\mathbb C))$, there exists four matrices, $X_{00}\in M_{d_0 k}(\mathbb C), X_{01}\in M_{d_0k\times d_1k}(\mathbb C), X_{10}\in M_{d_1k\times d_0k}(\mathbb C)$ and $X_{11}\in M_{d_1k}(\mathbb C)$ such that
\[X=\begin{bmatrix} X_{00}&X_{01}\\X_{10}&X_{11}\end{bmatrix}.\]
The stability of $\mathcal A$ algebras in either the Schr\"odinger or Heisenberg pictures imposes the same constraints on the bipartite unitary matrices.
\begin{theorem}\label{thm:Zero_bloc_unitary}
Let $U\in \mathcal U_{nk}$ and $\mathcal A$ be the above subalgebra of $M_n(\mathbb C)$. Then the following statements are equivalent:
\begin{enumerate}
\item There exists a set of quantum states $\mathcal B$ spanning $M_k(\mathbb C)$ such that for any $\beta \in\mathcal B$, the subalgebra $\mathcal A$ is stable under the UCP map $S_{U,\beta}$:
\[\forall \beta\in\mathcal B,\quad S_{U,\beta}(\mathcal A)\subset \mathcal A.\]
\item There exists a set of quantum states $\mathcal B$ spanning $M_k(\mathbb C)$ such that for any $\beta \in\mathcal B$, the subalgebra $\mathcal A$ is stable under the TPCP map $T_{U,\beta}$:
\[\forall \beta\in\mathcal B,\quad T_{U,\beta}(\mathcal A)\subset \mathcal A.\]
\item It holds that
\[U=\begin{bmatrix}U_{00}&0\\ 0& U_{11}\end{bmatrix}\]
with $U_{00}$ and $U_{11}$ unitary matrices.
\end{enumerate}
The implications (1)$\implies$(3) and (2)$\implies$(3) hold if the stability of the algebra $\mathcal A$ holds for a positive definite state $\beta$. In other words if $\mathcal B$ is replaced in respectively (1) and (2) by a singlet $\mathcal B=\{\beta\}$ with $\beta>0$.
\end{theorem}
\begin{proof}
We start with the implication (1)$\implies$(3). Statement (1) is equivalent to 
\[\operatorname{Tr}(XS_{U,\beta}(Y))=0\]
for any $X\in M_n(\mathbb C)$ positive semidefinite such that $X_{11}=0$ and $Y\in M_n(\mathbb C)$ positive semidefinite such that $Y_{00}=0$. From the definition of $S_{U,\beta}$ in \eqref{eq:def-UCP}, we get
\[\operatorname{Tr}(U^*Y\otimes I_k U X\otimes \beta)=0.\]
Since $\beta>0$ and, $X$ and $Y$ are positive semidefinite, there exists a constant $c>0$ such that,
\[\operatorname{Tr}(U^*Y\otimes I_k U X\otimes \beta)\geq c \operatorname{Tr}(U^*Y\otimes I_k U X\otimes I_k).\]
Thus, (1) implies,
\[\operatorname{Tr}(U^*Y\otimes I_k U X\otimes I_k)=0\]
for any $X\in M_n(\mathbb C)$ positive semidefinite such that $X_{11}=0$ and $Y\in M_n(\mathbb C)$ positive semidefinite such that $Y_{00}=0$. Choosing $X=P_{0}$ and $Y=P_1$ with $P_i$ the orthogonal projector on $V_i$,
\[\operatorname{Tr}(U^*P_1\otimes I_k U P_0\otimes I_k)=0.\]
Hence $\|U_{10}\|_{HS}=0$ where $\|\cdot\|_{HS}$ is the Hilbert--Schmidt norm. Since $U$ is a unitary matrix, $U_{01}$ implies $U_{11}U_{11}^*=I_{d_1}$ and $U_{00}^*U_{00}=I_{d_0}$. Hence both $U_{00}$ and $U_{11}$ are full rank and thus unitary matrices by Cochran's Lemma \ref{lem:Cochran}. Using again the fact that $U$ is unitary, $U_{11}U_{01}^*=0$ yields statement (3).

\medskip
The proof of (2)$\implies$(3) is similar. Statement (2) is equivalent to
\[\operatorname{Tr}(XT_{U,\beta}(Y))=0\]
for any $X\in M_n(\mathbb C)$ positive semi definite such that $X_{11}=0$ and $Y\in M_n(\mathbb C)$ positive semi definite such that $Y_{00}=0$. From the definition of $T_{U,\beta}$ in \eqref{eq:def-TPCP},
\[\operatorname{tr}(X\otimes I_kUY\otimes \beta U^*)=0.\]
Since $\beta>0$, taking $X=P_0$ and $Y=P_1$ we deduce, $U_{01}=0$. Using again the fact that $U$ is unitary, statement (3) follows.

\medskip
The proof of the implications (3)$\implies$(1) and (3)$\implies$(2) are straightforward once one remarks that (3) implies that $U(V_1\otimes\mathbb C^k)\subset V_1\otimes\mathbb C^k$.
\end{proof}
\begin{remark}
Imposing a finer subalgebra structure to $\mathcal A$ amounts to apply one of the appropriate theorems from previous sections to the unitary bloc $U_{11}$.
\end{remark}
\section{Generating random structured bipartite unitary operators}
\label{sec:algorithms}

We discuss in this section the natural probability measure one can consider on the sets of bipartite unitary operators appearing in Theorems \ref{thm:diagonal-Heisenberg},\ref{thm:diagonal-Schrodinger},\ref{thm:block-diagonal-Heisenberg},\ref{thm:block-diagonal-Schrodinger},\ref{thm:tensor-product-H},\ref{thm:tensor-product-S}. We shall discuss each problem separately, in the order of increasing difficulty. 

Let us start with the set of unitary operators described in Theorem \ref{thm:diagonal-Heisenberg}; these are bipartite operators $U=(U_{ij})_{i,j=1}^n$ with the property that the blocks $U_{ij} \in M_k(\mathbb C)$ are partial isometries satisfying conditions \ref{it:row-final},\ref{it:col-initial} from Definition \ref{def:bistochastic-PI}. Let us denote by $\delta_{ij} = \dim E_{ij}$. From condition \ref{it:col-initial}, we deduce that 
$$\forall j \in [n], \qquad \sum_{i=1}^n \delta_{ij} = k.$$
Since $U_{ij}$ is a partial isometry, we also have $\delta_{ij} = \dim F_{ij}$, and thus, from condition \ref{it:row-final}, we get
$$\forall i \in [n], \qquad \sum_{j=1}^n \delta_{ij} = k.$$
We conclude that the non-negative integer matrix $\delta = (\delta_{ij})_{i,j=1}^n$ has the property that each of its rows and columns forms a partition of $k$. For each such \emph{pattern matrix} $\delta$, we construct the following probability measure on the set of unitary matrices $U$. 

\medskip

\noindent\textbf{Random bipartite unitary operators satisfying  \ref{it:row-final},\ref{it:col-initial} with a given pattern}
\begin{enumerate}
\item \textbf{Input:} A pattern matrix $\delta$.
\item Consider $2n$ i.i.d., Haar distributed unitary random matrices $R_1, \ldots, R_n, C_1, \ldots, C_n \in \mathcal U_k$. Denote by $R_i(j)$ the $j$-th column of the matrix $R_i$; similarly for $C_i(j)$.
\item Define the input spaces $E_{ij}$ as follows:
$$\forall i,j \in [n], \qquad E_{ij} := \operatorname{span} \left\{C_j\left(\sum_{t<i}\delta_{tj} + s\right)\right\}_{s=1}^{\delta_{ij}},$$
and the output spaces $F_{ij}$ as 
$$\forall i,j \in [n], \qquad F_{ij} := \operatorname{span} \left\{R_i\left(\sum_{t<j}\delta_{it} + s\right)\right\}_{s=1}^{\delta_{ij}}.$$
\item For all $i,j \in [n]$, the partial isometry $U_{ij}$ is defined by its action on $E_{ij}$, the orthogonal of its kernel:
$$\forall s \in [\delta_{ij}], \qquad U_{ij} C_j\left(\sum_{t<i}\delta_{tj} + s\right) = R_i\left(\sum_{t<j}\delta_{it} + s\right).$$
\item \textbf{Output:} $U = \sum_{i,j=1}^n e_ie_j^* \otimes U_{ij}$, a random unitary matrix satisfying the conditions from Theorem \ref{thm:block-diagonal-Heisenberg}.
\end{enumerate}

\medskip

Note that the choice of the pattern matrix $\delta$ in the algorithm above is arbitrary. A canonical choice only appears when $k=n$, in which case one can choose $\delta_{ij}=1$, for all $i,j \in [n]$. The case of unitary operators appearing in Theorem \ref{thm:block-diagonal-Heisenberg} is dealt with in a similar manner, we leave the details to the reader. 

Let us now consider bipartite unitary operators that satisfy the conditions from Theorem \ref{thm:tensor-product-H}. Such operators admit a decomposition  $U = (I_d \otimes V) (W \otimes I_r)$ for other unitary matrices $V \in \mathcal U_{rk}$ and $W \in \mathcal U_{dk}$, and it is natural to choose $V$ and $W$ to be independent, Haar distributed in their respective group. Note that this probability measure does not give the full Haar measure on the group $\mathcal U_{drk}$, since the operators $V$ and $W$ ``interact'' only on the space $\mathbb C^k$. Counting parameters, we find that the real dimension of the set of $U$'s is $k^2(d^2+r^2-1)$. 

Regarding operators which give quantum channels preserving tensor product algebras (Theorem \ref{thm:tensor-product-S}), the situation is more complicated, since there is not an obvious way to sample uniformly from the set $\mathcal U_{unital}$ from \eqref{eq:Uunital}, see also \cite[Section 3]{dnp}.  We conjecture that the following algorithm will converge to an element of $\mathcal U_{unital}$. 

\medskip

\noindent\textbf{Sampling from $\mathcal U_{unital}$}
\begin{enumerate}
\item \textbf{Input:} Integers $n,k$ and an error parameter $\varepsilon >0$.
\item Start with a Haar distributed unitary random unitary operator $U \in \mathcal U_{nk}$.
\item While $\|U^\Gamma (U^\Gamma)^* - I_{nk} \|_2 > \varepsilon$, repeat the next step:
\item \qquad $U \leftarrow \operatorname{Pol}(U^\Gamma)$, where $\operatorname{Pol}(X)$ is the unitary operator $V$ appearing in the polar decomposition of $X$: $X = VP$ with $P \geq 0$.
\item \textbf{Output:} $U$, an operator at distance at most $\varepsilon$ from $\mathcal U_{unital}$.
\end{enumerate}

\medskip

In Figure \ref{fig:Uunital-numerics}, we present numerical evidence supporting our conjecture that the algorithm above converges. Note than an obstruction to the convergence of the algorithm would be an example of a matrix $U \in \mathcal U_{nk}$ such that
$$U^\Gamma \neq U = \operatorname{Pol}(U^\Gamma);$$
on such an input, the loop would be stuck on $U \notin \mathcal U_{unital}$. We do not know whether such matrices exist or not. An implementation of the algorithm above can be found at \cite{num}.

\begin{figure}[htbp] 
\includegraphics[width=0.4\textwidth]{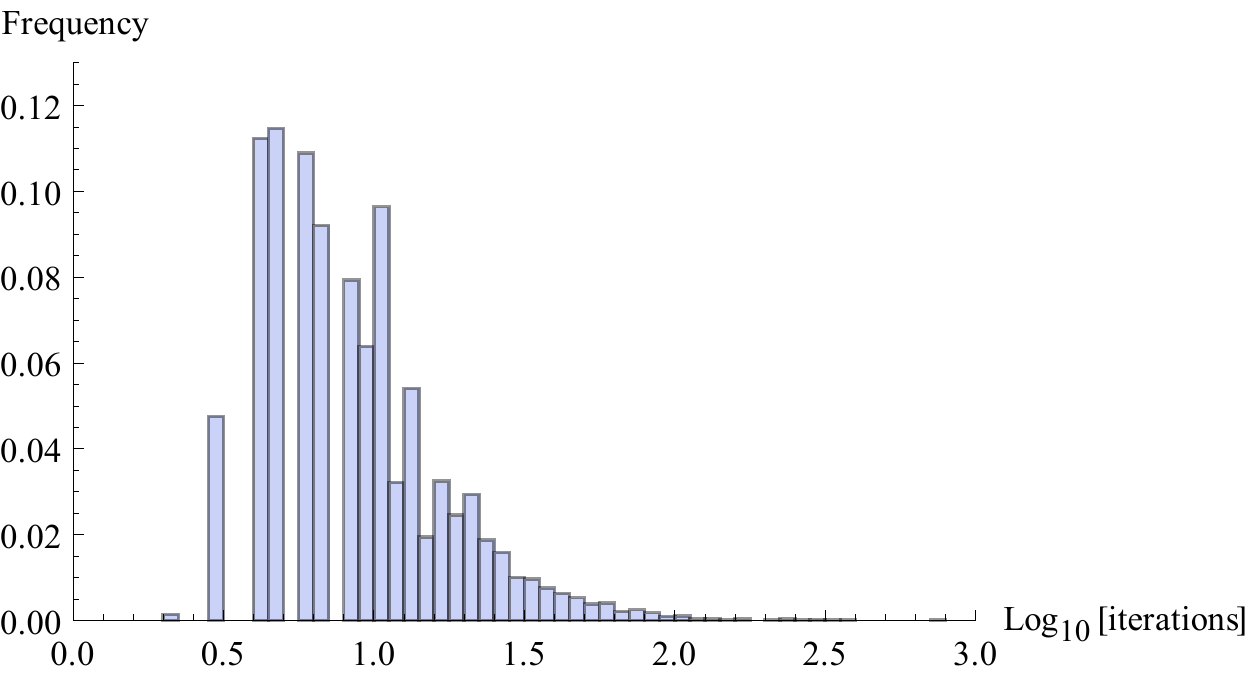} \qquad \includegraphics[width=0.4\textwidth]{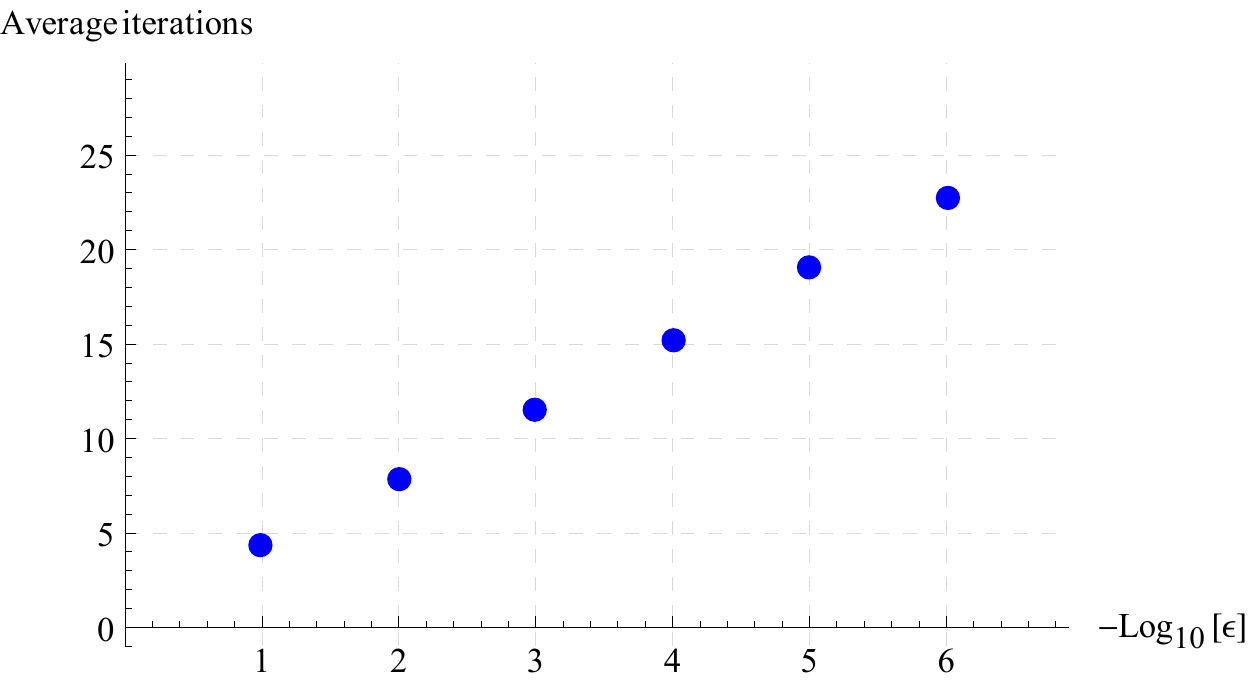}
\caption{Numerical evidence supporting the claim that algorithm that samples from $\mathcal U_{unital}$ converges. On the left panel, a histogram of the logarithm of the number of iterations that the algorithm needs to terminate for $n=k=2$ and $\varepsilon = 10^{-3}$, ran on $10^5$ random input data. On the right panel, for the same size parameters, the average number of steps the algorithm needs to converge, for different values of $\varepsilon$.} 
\label{fig:Uunital-numerics}
\end{figure}

Finally, let us discuss the case of unitary operators producing quantum channels which preserve diagonal subalgebras, see Theorem \ref{thm:diagonal-Schrodinger}. Generating such a matrix of partial isometries is very similar to what we have done in the Heisenberg case, with just one exception: the output subspaces need to partition $\mathbb C^k$ both row-wise and column-wise. Let us assume, for the sake of simplicity, that $n=k$ and $\dim F_{ij}=1$, for all $i,j \in [n]$; we denote by $x_{ij}$ the unit vector (up to a phase) which spans $F_{ij}$. We recover the following definition from \cite{mvi}. 

\begin{definition}\label{def:qls}
A \emph{quantum Latin square} (QLS) of order $n$ is a matrix $X = (x_{ij})_{i,j=1}^n$, where $x_{ij} \in \mathbb C^n$ are such that the vectors on each row (resp.~column) of $X$ form an orthonormal basis of $\mathbb C^n$. 
\end{definition}

We present next a conjectural algorithm for generating QLS. The main idea is to look at the projections on the vectors $x_{ij}$: $p_{ij} = x_{ij} x_{ij}^*$. Then, for some fixed column $j$, the fact that the vectors $\{x_{ij}\}_{i=1}^n$ form an orthonormal basis of $\mathbb C^n$ is equivalent to the following equality (a similar relation holds for the rows): $\sum_{i=1}^n p_{ij} = I_n$. Since this relation is similar to some probabilities summing up to one, we draw a parallel between QLS and bistochastic matrices. Our algorithm is thus an adaptation of the classical Sinkhorn-Knopp algorithm for generating bistochastic matrices \cite{sin,skn}. Although the procedure resembles Gurvits' algorithm \cite{gur} called ``operator scaling'', we have not been able to reduce it to operator scaling. Note also that the algorithm below also appears in the representation theory of the quantum symmetric group $S_N^+$, see \cite{bne}. For a $3$-tensor $X$, define the following $2n$ matrices, which have the vectors $x_{ij}$ as columns:
\begin{align*}
\forall i \in [n], \qquad R_i &= \sum_{j=1}^n x_{ij}e_j^*\\
\forall j \in [n], \qquad C_j &= \sum_{i=1}^n x_{ij}e_i^*.
\end{align*}
We say that $X$ is an $\varepsilon$-QLS if it is close to being a QLS, in the following sense:
$$\sum_{i=1}^n \|R_iR_i^* - I_n \|_2^2 + \sum_{j=1}^n \|C_jC_j^* - I_n\|_2^2 \leq \varepsilon^2.$$ 

\medskip

\noindent\textbf{Non-commutative Sinkhorn algorithm for sampling QLS}
\begin{enumerate}
\item \textbf{Input:} The dimension $n$ and an error parameter $\varepsilon >0$
\item Start with $x_{ij}$ independent uniform points on the unit sphere of $\mathbb C^n$.
\item While $X$ is not an $\varepsilon$-QLS, do the  steps (\ref{it:rows}-\ref{it:swap})
\item \label{it:rows}$\qquad$ Define the matrix $Y$ by making the rows of $X$ unitary:
$$\forall i \in [n], \qquad y_{ij} = \operatorname{Pol}\left( \sum_{s=1}^n x_{is}e_s^*  \right) \cdot e_j.$$
\item \label{it:cols}$\qquad$ Define the matrix $Z$ by making the rows of $Y$ unitary:
$$\forall j \in [n], \qquad z_{ij} = \operatorname{Pol}\left( \sum_{s=1}^n y_{sj}e_s^*  \right) \cdot e_i.$$
\item \label{it:swap}$\qquad$ $X \leftarrow Z$.
\item \textbf{Output:} $X$, an $\varepsilon$-QLS.
\end{enumerate}

\medskip

We conjecture that the algorithm above converges almost surely (for a computer implementation, see \cite{num}). Again, the proof of this important result is elusive at this time, but we have numerical evidence (see Figure \ref{fig:QLS-numerics}), as well as a proof of a relaxed version, where we replace the rank-1 projectors $p_{ij}$ with positive definite operators, see Appendix \ref{sec:app-Sinkhorn-positive}.

\begin{figure}[htbp] 
\includegraphics[width=0.4\textwidth]{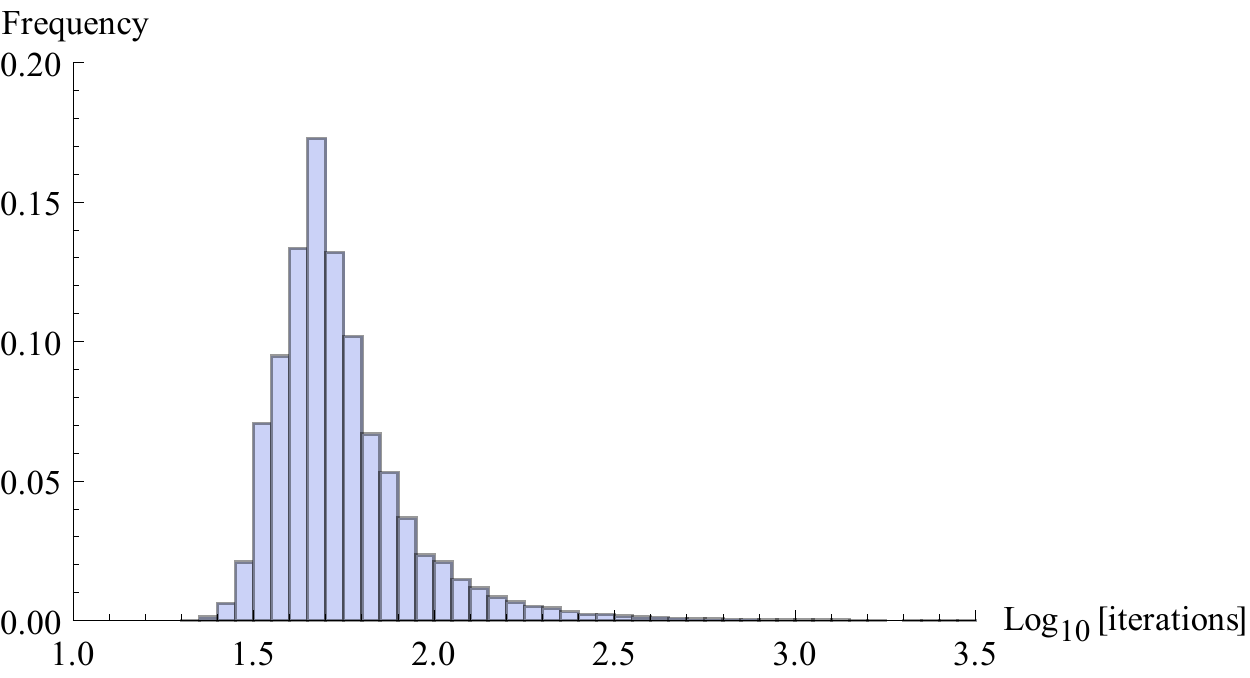} \qquad \includegraphics[width=0.4\textwidth]{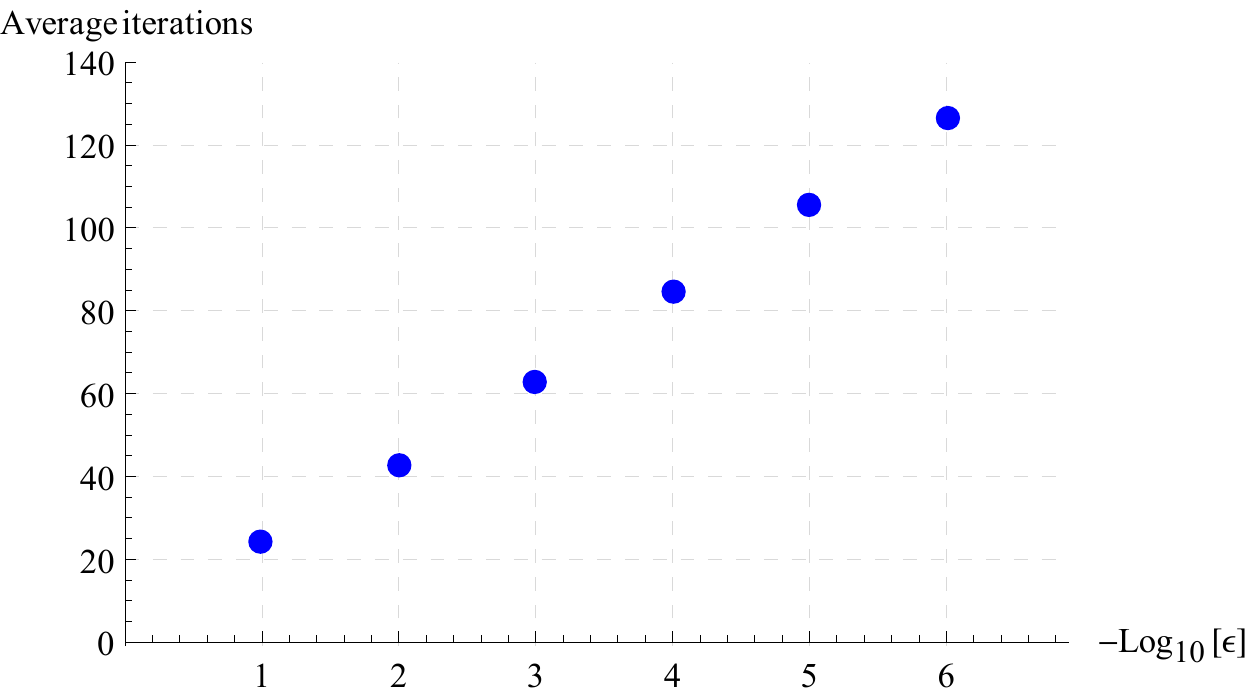}
\caption{Numerical evidence supporting the claim that algorithm that samples Quantum Latin Squares converges. On the left panel, a histogram of the logarithm of the number of iterations that the algorithm needs to terminate for $n=2$ and $\varepsilon = 10^{-3}$, ran on $10^5$ random input data. On the right panel, for the same size parameter, the average number of steps the algorithm needs to converge, for different values of $\varepsilon$.} 
\label{fig:QLS-numerics}
\end{figure}

\appendix
\section{A non-commutative Sinkhorn algorithm: the invertible case}\label{sec:app-Sinkhorn-positive}

We present in this section a non-commutative version of the classical Sinkhorn algorithm \cite{sin,skn}. The goal of the algorithm is to produce, starting from a matrix $X = (X_{ij})_{i,j=1}^n$ of positive definite blocks $M_k(\mathbb C) \ni X_{ij}>0$ another block matrix $Y=(Y_{ij})$ having the following two properties:
\begin{align*}
\forall i \in [n], \quad &\sum_{j=1}^n Y_{ij} = I_k\\
\forall j \in [n], \quad &\sum_{i=1}^n Y_{ij} = I_k.
\end{align*}
A block matrix having the two properties above will be called a \emph{block-bistochastic} matrix. As in the classical case, we shall need to define an approximate notion of block-bistochasticity: given some positive $\varepsilon > 0$, we call $Y$ $\varepsilon$-block-bistochastic if 
\begin{align*}
\forall i \in [n], \quad &\|I_k - \sum_{j=1}^n Y_{ij} \| \leq \varepsilon\\
\forall j \in [n], \quad &\|I_k - \sum_{i=1}^n Y_{ij} \| \leq \varepsilon.
\end{align*}

Let us introduce now the non-commutative variant of the classical Sinkhorn algorithm. 

\medskip

\noindent\textbf{Non-commutative Sinkhorn algorithm, the invertible case}
\begin{enumerate}
\item \textbf{Input:} A block matrix $X \in M_n(M_k(\mathbb C))$ with positive definite blocks $X_{ij} > 0$ and some positive precision parameter $\varepsilon >0$
\item Let $Y=X$
\item While $Y$ is not $\varepsilon$-block-bistochastic, do the  steps (\ref{it:step-rows}-\ref{it:step-swap})
\item \label{it:step-rows}$\qquad$ Define the block matrix $Y'$ by normalizing the rows of $Y$:
$$Y'_{ij} = \left(\sum_{s=1}^n Y_{is}\right)^{-1/2} Y_{ij} \left(\sum_{s=1}^n Y_{is}\right)^{-1/2} $$
\item \label{it:step-cols}$\qquad$ Define the block matrix $Y''$ by normalizing the columns of $Y'$:
$$Y''_{ij} = \left(\sum_{s=1}^n Y'_{sj}\right)^{-1/2} Y'_{ij} \left(\sum_{s=1}^n Y'_{sj}\right)^{-1/2} $$
\item \label{it:step-swap}$\qquad$ Let $Y = Y''$
\item \textbf{Output:} $Y$, a $\varepsilon$-block-bistochastic matrix.
\end{enumerate}

\medskip

Before proving that the algorithm above finishes after a finite number of steps, let us note that the inverses used in steps (\ref{it:step-rows}) and (\ref{it:step-cols}) are well-defined. Indeed, it is easy to check that at each step, the matrices $Y,Y',Y''$ have positive definite  blocks (we assume that the input $X$ has this property), and thus their row-sums and column-sums are also positive definite (and thus invertible) $k \times k$ matrices. 

\begin{proposition}
For any input $X$ having positive definite blocks and any precision parameter $\varepsilon>0$, the algorithm above stops after a finite number of steps.
\end{proposition}
\begin{proof}
Our proof is a natural extension of the proof in \cite{aar}. On the set $\mathcal X_{n,k}$ of $n \times n$ block matrices with strictly positive $k \times k$ blocks $X_{ij}$, we introduce the function $F : \mathcal X_{n,k} \to [0,\infty)$ given by
$$F(X) = \prod_{i,j=1}^n \det X_{ij}.$$

Let us now consider a matrix $Y \in \mathcal X_{n,k}$ with the property that the columns of $Y$ are normalized, i.e. $\forall j \in [n]$, $\sum_{i=1}^n Y_{ij} = I_k$. We claim that $F(Y) \leq 1$. Indeed, this is a consequence of the arithmetic-geometric (AG) mean inequality, as follows:
\begin{align*} 
F(Y) &= \prod_{i,j=1}^n \det Y_{ij} \leq \prod_{i,j=1}^n \left( \frac{\operatorname{Tr} Y_{ij}}{n} \right)^n\\
&\leq \left( \frac{1}{n^2} \sum_{i,j=1}^n \frac{1}{n} \operatorname{Tr} Y_{ij}\right)^{n^3}\\
&= \left( \frac{1}{n^3} \sum_{j=1}^n  \operatorname{Tr} \sum_{i=1}^n Y_{ij}\right)^{n^3}  = n^{-n^3} \leq 1.
\end{align*}

So, if we denote by $Y_1, Y_2, \ldots$ the matrices $Y$ which enter the algorithm at step (\ref{it:step-rows}), then we have $F(Y_r) \leq 1$, for all $r \geq 2$. We show now that, along the ``trajectory'' $Y_1, Y_2, \ldots$, the functional $F$ is increasing. We shall assume that $r \geq 2$ ($Y$ entering step (\ref{it:step-rows}) has normalized columns) and we shall prove that $F(Y) \leq F(Y')$. The inequality $F(Y') \leq F(Y'')$ can be proved in a similar manner. We have 
$$F(Y') = \prod_{i,j=1}^n \det Y'_{ij} =  \prod_{i,j=1}^n \frac{\det Y_{ij}}{\det \sum_{s=1}^n Y_{is}} = \frac{F(Y)}{\left( \prod_{i=1}^n \det \sum_{s=1}^n Y_{is}\right)^n},$$
and thus it is enough to show that
\begin{align} 
\label{eq:correction-factor-Sinkhorn-first} \left( \prod_{i=1}^n \det \sum_{s=1}^n Y_{is}  \right)^n &\leq \left(\prod_{i=1}^n \left( \frac{1}{n} \operatorname{Tr} \sum_{s=1}^n Y_{is}\right)^n\right)^n\\
\label{eq:correction-factor-Sinkhorn-second}  &\leq \left( \frac{1}{n} \sum_{i=1}^n \frac{1}{n} \sum_{s=1}^n \operatorname{Tr} Y_{is} \right)^{n^3} \\
&= \left( \frac{1}{n^2} \sum_{s=1}^n \operatorname{Tr} \sum_{i=1}^nY_{is} \right)^{n^3}= 1.
\end{align}
We have thus shown that the function $r \mapsto F(Y_r)$ is increasing and bounded along a run $Y_1, Y_2, \ldots$ of the algorithm, so it must converge to some value $f_\infty \in [0,1]$. In particular, the quantity appearing on the left hand side of \eqref{eq:correction-factor-Sinkhorn-first} converges to 1, as $r \to \infty$. As a consequence, the AG mean inequalities which were used in \eqref{eq:correction-factor-Sinkhorn-first} were almost equalities: given $\varepsilon$, there is some $R \geq 2$ such that, for all $r \geq R$ and all $i \in [n]$, 
$$(1-\varepsilon)\left( \frac{1}{n} \operatorname{Tr} \sum_{s=1}^n Y_{is}\right)^n\leq \det \sum_{s=1}^n Y_{is}  \leq \left( \frac{1}{n} \operatorname{Tr} \sum_{s=1}^n Y_{is}\right)^n.$$
It is a classical fact (see, e.g.~\cite{kob}) that this implies that that the numbers appearing in the AG inequality should be, up to some constants depending on $n$, $\varepsilon$-close to their mean; in our case this translates to the matrices $\sum_{i=1}^s Y_{is}$ being close to a multiple of the identity, which is our $\varepsilon$-block-bistochastic condition.
\end{proof}


\begin{thebibliography}{99}

\bibitem{aar}
Aaronson, S.
{\it Quantum computing and hidden variables.}
Phys. Rev. A, 71(3), 032325 (2005).

\bibitem{adler}
Adler, S. L., Brody, D. C., Brun T. A., Hughston L. P.
{\it Martingale models for quantum state reduction.}
J. Phys. A 34, 8795 (2001) 

\bibitem{bne}
Banica, T, Nechita, I.
{\it Universal models for quantum permutation groups.}
Preprint arXiv:1602.04456.

\bibitem{bar}
Bardet, I.
{\it Quantum extensions of dynamical systems and of Markov semigroups.}
Preprint arXiv:1509.04849.

\bibitem{num}
Benoist, T., Nechita, I.
{\it Computer implementation of some algorithms appearing in the current paper.}
See supplementary materials.

\bibitem{bb}
Bauer, M., Bernard, D.
{\it Convergence of repeated quantum nondemolition measurements and wave-function collapse.}
Phys. Rev. A 84, 044103 (2011)

\bibitem{bbb}
Bauer M, Benoist T and Bernard D 
{\it Repeated quantum non-demolition measurements: convergence and continuous-time limit}
Ann. H. Poincar\'e 14, 639--679 (2013)

\bibitem{bpe}
Benoist, T., Pellegrini, C.
{\it Large Time Behavior and Convergence Rate for Quantum Filters under Standard Non Demolition Conditions.}
Comm. Math. Phys. 331, 703--723 (2014)

\bibitem{coc}
Cochran, W. G.
{\it The distribution of quadratic forms in a normal system, with applications to the analysis of covariance.}
Math. Proc. Cambridge 30, 178--191 (1934).

\bibitem{dav}
Davidson, K.R.
{\it $C^*$-algebras by example.}
Am. Math. Soc. (1996).

\bibitem{dnp}
Deschamps, J., Nechita, I., Pellegrini, C.
{\it On some classes of bipartite unitary operators.}
J. Phys. A 49, 335301 (2016).

\bibitem{gur}
Gurvits, L.
{\it Classical complexity and quantum entanglement.}
Journal of Computer and System Sciences 69.3, pp 448--484 (2004).

\bibitem{kob}
Kober, H.
{\it On the arithmetic and geometric means and on H\"older's inequality.}
P. Am. Math. Soc. 9, 452--459 (1958).

\bibitem{mvi}
Musto, B., Vicary, J.
{\it Quantum Latin squares and unitary error bases.}
Preprint arXiv:1504.02715.

\bibitem{nch}
Nielsen, M., and Chuang, I. 
{\it Quantum computation and quantum information.} 
Cambridge University Press, 2010.

\bibitem{pau}
Paulsen, V.
{\it Completely bounded maps and operator algebras.}
Cambridge University Press (2002).

\bibitem{sin}
Sinkhorn, R.
{\it A relationship between arbitrary positive matrices and doubly stochastic matrices.}
Ann. Math. Stat., 35 876--879 (1964).

\bibitem{skn}
Sinkhorn, R. and Knopp, P.
{\it Concerning nonnegative matrices and doubly stochastic matrices.}
Pac. J. Math. 21, 343--348 (1967).

\bibitem{sti}
Stinespring, W. F. 
{\it Positive functions on $C^*$-algebras.}
P. Am. Math. Soc. 6, 211--216 (1955).

\bibitem{stockton}
Stockton, J., van Handel, R., Mabuchi, H.
{\it Deterministic Dick-state preparation with continuous measurement and control.}
Phys. Rev. A 70, 022106 (2004)

\bibitem{vanhandel}
van Handel, R., Stockton, J., Mabuchi, H.
{\it Feedback control of quantum state reduction.}
IEEE Trans. Autom. Control 50, 768 (2005)

\bibitem{wol}
Wolf, M. 
{\it Quantum channels \& operations: Guided tour.}
Lecture notes available  \href{http://www-m5.ma.tum.de/foswiki/pub/M5/Allgemeines/MichaelWolf/QChannelLecture.pdf}{online}, July 2012.

\end{thebibliography}
\end{document}